\documentclass[11pt]{article}
\usepackage[sc]{mathpazo}

\usepackage[margin=1in]{geometry}
\usepackage[english]{babel}
\usepackage[utf8x]{inputenc}
\usepackage[compact]{titlesec}
\usepackage{cmap}
\usepackage[T1]{fontenc}
\usepackage{bm}
\usepackage{booktabs}
\usepackage{mathtools}
\usepackage{amsfonts}
\usepackage{amsmath}
\usepackage{amssymb}
\usepackage{amsthm}
\usepackage{float}

\usepackage[section]{placeins} 

\usepackage{graphics}
\usepackage{hyperref}
\usepackage[svgnames]{xcolor}
\hypersetup{colorlinks={true},urlcolor={blue},linkcolor={DarkBlue},citecolor=[named]{DarkGreen},linktoc=all}
\usepackage{natbib}
\usepackage{microtype}
\usepackage[capitalise,nameinlink,noabbrev]{cleveref}
\usepackage{doi}

\newcommand{\BibTeX}{\rm B\kern-.05em{\sc i\kern-.025em b}\kern-.08em\TeX}

\usepackage[labelfont={normalfont,bf},textfont=it]{caption}
\usepackage{subcaption}

\usepackage{nicefrac}
\usepackage{pifont}

\usepackage{apxproof}
\usepackage{array,multirow,graphicx,bigdelim}

\input{insbox}
\usepackage{tikz}
\usepackage{tikz-dimline}
\usepackage{tikz-cd}
\usetikzlibrary{fit,calc,math,shapes,shapes.multipart,decorations.text,arrows,decorations.markings,decorations.pathmorphing,shapes.geometric,positioning,decorations.pathreplacing}
\tikzset{snake it/.style={decorate, decoration=snake}}

\colorlet{mygray}{gray!40}

\makeatletter
\let\oldnl\nl
\newcommand{\nonl}{\renewcommand{\nl}{\let\nl\oldnl}}
\makeatother

\usepackage{thmtools,thm-restate}

\DeclarePairedDelimiter{\paren}{\lparen}{\rparen}

\DeclarePairedDelimiter{\bracket}{[}{]}
\DeclarePairedDelimiter{\oc}{(}{]}

\DeclarePairedDelimiter{\set}{\lbrace}{\rbrace}
\DeclarePairedDelimiter{\card}{\lvert}{\rvert}

\DeclarePairedDelimiter{\abs}{\lvert}{\rvert}

\DeclareMathOperator*{\argmax}{arg\,max}
\DeclareMathOperator*{\argmin}{arg\,min}

\def\1{\mathds{1}}

\def\maxconf{\mathsf{max}\text{-}\mathsf{conf}}
\def\minconf{\mathsf{min}\text{-}\mathsf{conf}}

\newcommand{\A}{\mathcal{A}}

\newcommand{\EF}[1]{\ifstrempty{#1}{\textrm{\textup{EF}}}{\textrm{\textup{EF{$#1$}}}}}
\newcommand{\EFoneone}{\textrm{\textup{EF$[1,1]$}}}
\newcommand{\EFX}{\textrm{\textup{EFX}}}

\newcommand{\FairIS}{\textrm{\textup{\textsc{Fair}-IS}}}
\newcommand{\I}{\mathcal{I}}

\newcommand{\NPH}{\textrm{\textup{NP-hard}}}

\newcommand{\PO}{\textup{PO}}

\newcommand{\PPAC}{\textrm{\textup{PPA-complete}}}

\definecolor{DarkGreen}{rgb}{0.1,0.5,0.1}

\newcommand{\cA}{\mathcal{A}}

\newcommand{\N}{\mathbb{N}}

\newcommand{\tI}{\tilde{I}}

\newcommand{\tM}{\tilde{M}}
\newcommand{\tE}{\tilde{E}}
\newcommand{\tG}{\tilde{G}}
\newcommand{\tv}{\tilde{v}}
\newcommand{\tA}{\tilde{A}}
\newcommand{\tcA}{\tilde{\mathcal{A}}}

\theoremstyle{remark}

\newtheorem{example}{Example}

\title{Fair Allocation under Conflict Constraints}
\author{
	\begin{tabular}{m{0.12\textwidth}m{0.12\textwidth}m{0.12\textwidth}m{0.12\textwidth}m{0.12\textwidth}m{0.12\textwidth}}
		\multicolumn{2}{c}{\textbf{Sarfaraz Equbal}} & \multicolumn{2}{c}{\textbf{Rohit Gurjar}} & \multicolumn{2}{c}{\textbf{Ayumi Igarashi}}\\
		\multicolumn{2}{c}{\small{IIT Bombay}} & \multicolumn{2}{c}{\small{IIT Bombay}} & \multicolumn{2}{c}{\small{The University of Tokyo}}\\
		\multicolumn{2}{c}{\href{mailto:sequbal@cse.iitb.ac.in}{\small{\texttt{sequbal@cse.iitb.ac.in}}}} & \multicolumn{2}{c}{\href{mailto:rgurjar@cse.iitb.ac.in}{\small{\texttt{rgurjar@cse.iitb.ac.in}}}}
		& \multicolumn{2}{c}{\href{mailto:igarashi@mist.i.u-tokyo.ac.jp}{\small{\texttt{igarashi@mist.i.u-tokyo.ac.jp}}}}\\
		&&&&&\\
		\multicolumn{2}{c}{\textbf{Yatharth Kumar}} & 
        \multicolumn{2}{c}{\textbf{Pasin Manurangsi}} & 
        \multicolumn{2}{c}{\textbf{Swaprava Nath}}\\
		\multicolumn{2}{c}{\small{Rubrik}} &
        \multicolumn{2}{c}{\small{Google Research}} &
        \multicolumn{2}{c}{\small{IIT Bombay}}\\
		\multicolumn{2}{c}{\href{mailto:yatharthbns811@gmail.com}{\small{\texttt{yatharthbns811@gmail.com}}}} &
        \multicolumn{2}{c}{\href{mailto:pasin@google.com}{\small{\texttt{pasin@google.com}}}} &
        \multicolumn{2}{c}{\href{mailto:swaprava@cse.iitb.ac.in}{\small{\texttt{swaprava@cse.iitb.ac.in}}}}\\
        &&&&&\\
		\multicolumn{2}{c}{\textbf{Raghuvansh Saxena}} &
        \multicolumn{2}{c}{\textbf{Rohit Vaish}} &
        \multicolumn{2}{c}{\textbf{Hirotaka Yoneda}}\\
		\multicolumn{2}{c}{\small{TIFR}} &
        \multicolumn{2}{c}{\small{IIT Delhi}} &
        \multicolumn{2}{c}{\small{The University of Tokyo}}\\
		\multicolumn{2}{c}{\href{mailto:raghuvansh.saxena@tifr.res.in}{\small{\texttt{raghuvansh.saxena@tifr.res.in}}}} &
        \multicolumn{2}{c}{\href{mailto:rvaish@iitd.ac.in}{\small{\texttt{rvaish@iitd.ac.in}}}} &
        \multicolumn{2}{c}{\href{mailto:yoneda-h@g.ecc.u-tokyo.ac.jp}{\small{\texttt{yoneda-h@g.ecc.u-tokyo.ac.jp}}}}\\
	\end{tabular}
}

\usepackage{colortbl}
\colorlet{myred}{red!25}
\colorlet{myblue}{blue!25}
\colorlet{mygreen}{green!25}

\usepackage{todonotes}

\usepackage[shortlabels]{enumitem}

\usepackage{etoolbox,algorithm,algpseudocode}

\makeatletter
\newcommand{\algmargin}{\the\ALG@thistlm}
\algnewcommand{\parState}[1]{\State%
  \parbox[t]{\dimexpr\linewidth-\algmargin}{\strut #1\strut}}

\newcommand{\alglinenoNew}[1]{\newcounter{ALG@line@#1}}
\newcommand{\alglinenoPop}[1]{\setcounter{ALG@line}{\value{ALG@line@#1}}}
\newcommand{\alglinenoPush}[1]{\setcounter{ALG@line@#1}{\value{ALG@line}}}

\newtheorem{theorem}{Theorem}[section]
\newtheorem*{theorem*}{Theorem}

\newtheorem{observation}[theorem]{Observation}
\newtheorem*{observation*}{Observation}
\newtheorem{lemma}{Lemma}
\newtheorem{proposition}{Proposition}
\newtheorem{definition}{Definition}

\date{}

\begin{document}

\maketitle 


\begin{abstract}
  We study the fair allocation of indivisible items subject to conflict constraints. In this framework, the items are represented as the vertices of a graph, with edges corresponding to conflicts between pairs of items. Each agent is assigned an independent set of items from the graph. Our goal is to achieve a fair and efficient allocation of these items. Fairness pertains to satisfying envy-freeness up to one item~(\EF{1}), while efficiency is defined by maximality, meaning that no unallocated item can be feasibly assigned to any agent.

First, we explore the case of two agents. For monotone valuations, we show that a maximal \EF{1} allocation always exists on any graph. Our existence proof relies on a color-switching technique, which locally modifies a maximal allocation while preserving feasibility and restoring \EF{1}. We further show that such allocations can be computed in pseudopolynomial time in general, and in polynomial time for additive valuations on arbitrary graphs, as well as for monotone valuations on interval and bipartite graphs. By contrast, once monotonicity is dropped, maximal \EF{1} allocations need not exist even for identical additive valuations, and deciding existence becomes \NPH{}.

Next, we consider the case with a general number of agents. Again, we arrive at a negative result: An \EF{1} and maximal allocation fails to exist even for three agents under identical monotone valuations, and determining the existence of such an allocation is NP-hard. On the positive side, we show that under identical non-monotone additive valuations on a path graph, an \EFoneone{} and maximal allocation always exists. This result involves a novel application of the ``cycle plus triangles'' theorem.\footnote{Preliminary versions of this work have appeared at the AAMAS 2024~\citep{KEG+24} and IJCAI 2025~\citep{IMY25dividing} conferences. }

\end{abstract}

\section{Introduction}
\label{sec:Introduction}



How can we allocate a set of resources fairly? This problem was first formalized by the pioneering work of~\citet{Steinhaus48} and has since been extensively studied in economics, mathematics, and computer science under the umbrella of \emph{fair division}~\citep{BT96fair,BCE+16handbook,M19fair}. Applications of fair division arise in many real-life situations, including 
the allocation of courses among students~\citep{BudishCaKe17}, the division of family inheritance among family members~\citep{GoldmanPr14}, and the division of household chores between couples~\citep{IgarashiYo23}.

A central concept of fairness in fair division is \emph{envy-freeness}, which requires that each agent receives their most preferred bundle in the allocation. Although envy-freeness is a natural and well-motivated fairness guarantee, allocations that meet this criterion could fail to exist when dealing with indivisible resources, such as courses, houses, or tasks. As a result, recent literature on discrete fair division has focused on approximate fairness, exploring various concepts and algorithmic results~\citep{AAB+23fair}. One particularly influential relaxation of envy-freeness is \emph{envy-freeness up to one item} (\EF{1}), introduced by \citet{B11combinatorial}. This notion allows agents to eliminate pairwise envy by (hypothetically) dropping one of the items. This concept has garnered significant attention over the past decade. It is known that for general classes of monotone valuations, an \EF{1} allocation exists and can be computed efficiently~\citep{LMM+04approximately,BSV21approximate}.

A common assumption in the fair division literature is that any item can be feasibly assigned to any agent. This assumption may not hold in many settings of interest. For example, in course allocation, a student can only attend at most one course at any given time. Similarly, when allocating players to sports teams, if two players have overlapping areas of expertise, it is preferable not to assign them to the same team. In such settings, it is more natural to model \emph{conflicts} among the items and allow only feasible (or non-conflicting) allocations. 
A versatile framework for modeling such constraints, explored in a series of recent papers~\citep{ChiarelliKMPPS20,HH22fair}, represents conflicts among indivisible resources as a graph, where the vertices correspond to resources and the edges represent conflicts.

\paragraph{Our contributions.}
In this paper, we study the allocation of indivisible items under conflict constraints.
The items are represented by the vertices of a graph, and each agent must receive an
independent set of items. We consider a general valuation setting in which an item may
provide either non-negative utility, as a good, or negative utility, as a chore. Conflict
constraints substantially change the interaction between fairness and efficiency, and
standard guarantees and techniques from unconstrained fair division need not apply.

In particular, canonical efficiency notions such as completeness and Pareto optimality
are not always compatible with the fairness requirement of \EF{1} in the presence of conflict constraints. A complete
allocation assigning all goods may fail to exist, and even when it does,
there are simple instances in which no complete \EF{1} allocation can be achieved
\citep{HH22fair}. Pareto optimality is also too demanding for our purposes: even when the conflict graph is a path, \EF{1} is incompatible with Pareto optimality for two agents with identical additive
valuations (see \Cref{fig:EF1_PO_line_graph_counterexample}), where Pareto optimality is defined with respect to feasible allocations. These limitations motivate our focus on the weaker efficiency notion of \emph{maximality}, which requires that no unallocated item can be feasibly assigned to any agent.

Besides \EF{1}, another natural benchmark is envy-freeness up to any item (\EFX{})~\citep{CKM+19unreasonable}. Unfortunately, \EFX{} and maximality are incompatible: An allocation satisfying both \EFX{} and maximality may fail to exist even for two agents with identical valuations on a path graph
(see \Cref{eg:EFX_Maximal_Counterexample}).
We present the incompatibility results between fairness and efficiency notions in \Cref{fig:relations}.
Our goal is to identify when maximal \EF{1} allocations exist and when they can be
computed efficiently.

Our main contributions are as follows (see~\Cref{tab:Summary}).

\begin{enumerate}
    \item {\bf Two agents.}

    The two-agent problem is particularly important in fair division, with applications such as inheritance division, house-chore division, and divorce settlements~\citep{BramsFi00,brams2014two,IgarashiYo23}. We prove that, for two agents with monotone valuations, a maximal \EF{1} allocation exists for every conflict graph. The proof is based on a ``color switching'' technique, which constructs a sequence of maximal allocations such that the bundles of the agents between consecutive allocations satisfy an adjacency condition. It can be shown that, in any pair of consecutive allocations where the sign of envy between the agents changes, at least one allocation must satisfy \EF{1}.

    We complement this structural result with algorithms. For additive valuations, we give
    a polynomial-time algorithm for finding a maximal \EF{1} allocation on arbitrary
    graphs. For general monotone valuations, we obtain a pseudo-polynomial-time algorithm.
    We also identify graph classes for which efficient computation is possible under
    monotone valuations, including interval graphs, a structure that naturally arises in
    scheduling applications. When monotonicity is dropped, however, the existence is no longer guaranteed: there is a simple example in which a maximal \EF{1} allocation fails to exist even for two agents with identical additive valuations.

    \item {\bf General number of agents.}
    
    We show that the positive two-agent result does not extend to an arbitrary number of
    agents. We show multiple families of counterexamples for \EF{1} and maximality, including for (a) any $n \geq 3$ agents with identical monotone valuations, (b) any $n \geq 4$ agents with identical additive and monotone valuations, and (c) any $n \geq 2$ agents with identical additive valuations. 
    Furthermore, we provide a general recipe for converting any counterexample to an \NPH{}ness proof. In particular, we show that checking the existence of \EF{1} and maximal allocations for all of the aforementioned settings is \NPH{}.

    On the positive side, we achieve a meaningful relaxation of \EF{1}, called envy-freeness up to one good and one chore (\EFoneone{}), for path graphs. This notion addresses pairwise envy by removing one good from the envied bundle and one chore from the envious agent's bundle~\citep{ShoshanSeHa23}. For any
    number of agents with identical additive valuations on a path, we prove that a maximal
    \EF{[1,1]} allocation always exists. The proof employs a novel application of the ``cycle-plus-triangles'' theorem from graph theory~\citep{FS92solution,FS97some}, which states that any graph whose edge set is a disjoint union of triangles and a Hamiltonian cycle is $3$-colorable. The \EF{1} and maximal allocation resulting from our proof can be achieved by a generalized round-robin algorithm, where different rounds use different permutations of agents, and, on its turn, each agent picks its favorite remaining item subject to path-conflict constraints.
    Thus, although an \EF{1} and maximal allocation may fail to exist for general graphs, a robust relaxation of \EF{1}, namely \EFoneone{}, remains attainable on paths. Finally, we consider the case of \emph{uniform valuations} in which each item has the same value. We show that when the conflict graph is a tree, a maximal \EF{1} allocation is guaranteed to exist and can be found in polynomial time. 
\end{enumerate}

\paragraph{Related work.}
There is a growing body of research on fair division under constraints. For a comprehensive survey on the topic, see \citep{Suksompong21}. The concept of conflict constraints in fair division was introduced by~\citet{ChiarelliKMPPS20} and has been further explored in several subsequent works~\citep{HH22fair,KEG+24,BiswasFORC2023,LiFairScheduling2021}. In their study,~\citet{ChiarelliKMPPS20} explored different fairness objectives than those we consider, focusing on partial allocations that maximize the egalitarian social welfare (i.e., the utility of the worst-off agent). \citet{HH22fair} studied complete allocations satisfying fairness criteria such as EF1 and MMS (maximin fair share). 
\citet{BiswasFORC2023} generalized the model of \citet{HH22fair} by incorporating capacity of resources into their analysis.  
Additionally,~\citet{LiFairScheduling2021} considered an interval scheduling problem, with the goal of achieving fairness concepts such as EF1 and MMS.

Assigning independent subsets to agents naturally corresponds to a partial coloring of the conflict graph. A seminal result in this line of work is the Hajnal-Szemer\'{e}di theorem~\citep{HS70}, which states that any graph with maximum degree $\Delta$ admits an equitable
coloring when there are at least $\Delta+1$ colors.\footnote{An equitable coloring is a proper coloring in which the color classes are almost balanced. That is, each vertex is
assigned a color such that no pair of adjacent vertices have the same color and any two color classes differ in size by at
most one.} The distinguishing point with our work is that this result only focuses on equalizing the \emph{number} of vertices in each color class, whereas our model accounts for the (possibly differing) \emph{valuations} assigned by the agents to the vertices.

A related class of constraints is that of connectivity constraints in a graph~\citep{BouveretCeEl17,BiloCaFl22}, where each agent receives a connected bundle of a graph. Note that while connectivity constraints imposed by a sparse graph, such as a tree, generally result in fewer feasible allocations, in our context, sparsity actually leads to greater flexibility by increasing the number of feasible allocations.

\begin{table}[t]
\centering
\small
\begin{tabular}{|c|c|c|>{\centering\arraybackslash}p{2.25cm}|>{\centering\arraybackslash}p{5cm}|c|}
 \hline
 $\bm{n}$ & \textbf{Valuations} & \textbf{Graphs} & \textbf{Result} & \textbf{Notes} & \S\\
 \hline
 \rowcolor{myblue}
 \raisebox{-\height}{$2$} & \raisebox{-\height}{Monotone} & \raisebox{-\height}{General} & \raisebox{-\height}{Exists} & Efficient for bipartite/interval graphs or additive valuations & \raisebox{-\height}{\ref{sec:Results_Two_Agents}} \\
 \hline
 \rowcolor{myred}
 \raisebox{-\height}{$2$} & \raisebox{-\height}{Non-monotone} & \raisebox{-\height}{General} & Fails to exist and \NPH{} & \raisebox{-\height}{Even for identical additive valuations} & \raisebox{-\height}{\ref{sec:Non-Existence-Results}} \\ 
 \hline
 \rowcolor{myblue}
 \raisebox{-\height}{$>2$} & \raisebox{-\height}{Identical additive} & \raisebox{-\height}{Paths} & \raisebox{-\height}{Exists} & \EF{1} for monotone and \EFoneone{} for non-monotone & \raisebox{-\height}{\ref{sec:Results_n_Agents}} \\
 \hline
 \rowcolor{myred}
 \raisebox{-\height}{$>2$} & \raisebox{-\height}{Identical monotone} & \raisebox{-\height}{General} & Fails to exist and \NPH{} & Even for identical additive valuations if $n > 3$ & \raisebox{-\height}{\ref{sec:Non-Existence-Results}} \\
 \hline
\end{tabular}
\vspace{0.1in}
\caption{Summary of our results for \EF{1} and maximality. Each row presents an existential and/or computational result based on the assumptions specified in the corresponding columns. Rows highlighted in red indicate hardness or non-existence results, while blue rows represent existence results.}
\label{tab:Summary}
\end{table}

\section{Preliminaries}
\label{sec:Preliminaries}

Given any $r \in \mathbb{N}$, let $[r] \coloneqq \{1,2,\dots,r\}$. 

\paragraph{Problem instance.}
We use $M = \{o_1,o_2,\dots,o_m\}$ to denote the set of $m$ \emph{items} and $N = [n]$ to denote the set of $n$ \emph{agents}. Let $G = (M, E)$ denote an undirected graph, where each vertex corresponds to an item, and each edge corresponds to a conflict. We refer to $G$ as a \emph{conflict graph}. Each agent $i$ has a \emph{valuation function} $v_i: 2^M \rightarrow \mathbb{R}$; here,
$\mathbb{R}$ is the set of reals. 
We assume that $v_i(\emptyset) = 0$. 
For simplicity, the valuation of a single item $o \in M$, namely $v_i(\{o\})$, is denoted by $v_i(o)$. An instance of our problem is given by the tuple $(N, M, V, G)$ where $V = (v_1, v_2, \dots, v_n)$ denotes a valuation profile. We use $K_{s,t}$ to denote a complete bipartite graph with a bipartition in which one part contains $s$ vertices and another part includes $t$ vertices.

\paragraph{Classes of valuation functions.}
A valuation profile $\mathcal{V}$ is called \emph{monotone} if either all agents' valuations are monotone non-decreasing, or all agents' valuations are monotone non-increasing. A valuation function $v_i$ is \emph{monotone non-decreasing} if $v_i(S) \leq v_i(T)$ for every $S \subseteq T \subseteq M$. 
It is \emph{monotone non-increasing} if $v_i(S) \geq v_i(T)$ for every $S \subseteq T \subseteq M$. 
It is \emph{additive} if $v_i(S)=\sum_{o \in S} v_i(o)$ holds for every $S \subseteq M$ and $i\in N$. An instance $(N, M, V, G)$ is said to be \emph{of goods} (respectively,~\emph{of chores}) if each agent $i$ has a monotone non-decreasing (respectively,~monotone non-increasing) valuation function $v_i$. Additionally, a valuation profile $\mathcal{V}$ is called \emph{identical} if every agent $i \in N$ has the same valuation function; in this case, we denote this function by $v$. 
Our algorithmic results will assume access to a \emph{value query} oracle. In this model, for any given subset of items $S \subseteq M$ and any agent $i \in N$, a value query returns agent $i$'s value for the subset $S$, namely~$v_i(S)$.

\paragraph{Allocation.}
An \emph{allocation} is an ordered subpartition $\mathcal{A}=(A_1,\ldots, A_n)$ of $M$ where for every pair of distinct agents $i,j \in N$, $A_i \cap A_j = \emptyset$, $\cup_{i \in N} A_i \subseteq M$, and for each $i \in N$, $A_i$ is an \emph{independent set} of $G$, namely, there is no pair of items in $A_i$ that are adjacent to each other. The subset $A_i$ is called the \emph{bundle} of agent~$i$. An allocation is \emph{complete} if all items are allocated, i.e., $\cup_{i\in N} A_i=M$.

\paragraph{Fairness notions.}
An allocation is \emph{envy-free} if for every pair of agents $i,j \in N$, we have $v_i(A_i) \geq v_i(A_j)$~\citep{GS58puzzle,F67resource}. 
It is called \emph{envy-free up to one item} (\EF{1}) if for every pair of agents $i,j \in N$, there exists a set $S \subseteq A_i \cup A_{j}$ of size at most $1$ such that $v_i(A_i \setminus S) \ge v_i(A_{j} \setminus S)$~\citep{ACI+19fair}. 
An allocation is \emph{envy-free up to any item} (\EFX{}) \citep{aziz2022fair,AleksandrovWalsh2020MixedManna} if for every pair of agents $i,j \in N$, the following two conditions hold:
\begin{enumerate}[(i)]
    \item for every item $o \in A_j$ such that $v_i(A_j \setminus \{o\}) < v_i(A_j)$, we have $v_i(A_i) \geq v_i(A_j \setminus \{o\})$; and
    \item for every item $o \in A_i$ such that $v_i(A_i \setminus \{o\}) > v_i(A_i)$, we have $v_i(A_i \setminus \{o\}) \geq v_i(A_j)$.
\end{enumerate}
We consider a further relaxation of envy-freeness up to one item, denoted by EF$[1,1]$~\citep{ShoshanSeHa23}.
An allocation is \emph{envy-free up to one good and one chore} (EF$[1,1]$) if for every pair of agents $i,j \in N$, there exist sets $S \subseteq A_i$ and $T \subseteq A_{j}$, each of size at most $1$, such that $v_i(A_i \setminus S) \ge v_i(A_{j} \setminus T)$.

Note that for the goods-only problem, both \EF{1} and EF$[1,1]$ coincide with the standard definition of \EF{1} for goods. An allocation $\mathcal{A} = (A_1, \dots, A_n)$ is \emph{envy-free up to one good} (EF1 for goods) if for every pair of agents $i,j \in N$, either $A_j = \emptyset$, or $v_i(A_i) \geq v_i(A_j  \setminus \{o\})$ for some $o \in A_j$~\citep{B11combinatorial,LMM+04approximately}. An analogous equivalence also holds for chore instances. In this case, an allocation $\mathcal{A} = (A_1, \dots, A_n)$ is \emph{envy-free up to one chore} (EF1 for chores) if for every pair of agents $i,j \in N$, either $A_i = \emptyset$, or $v_i(A_i \setminus \{o\}) \geq v_i(A_j)$ for some $o \in A_i$~\citep{ACI+19fair,BSV21approximate}.

\paragraph{Efficiency notions.}
One of the fundamental notions of efficiency in fair division is \emph{completeness}, which requires that all items be allocated. Formally, an allocation $\mathcal{A}=(A_1,\ldots, A_n)$ is \emph{complete} if $\bigcup_{i \in N}A_i=M$.
Besides completeness, another commonly used notion of efficiency in fair division is \emph{Pareto optimality}. 
Specifically, an allocation $\mathcal{A}'$ \emph{Pareto dominates} another allocation $\mathcal{A}$ if $v_i(A'_i) \geq v_i(A_i)$ for every $i \in N$ and the inequality is strict for at least one agent. An allocation $\mathcal{A}$ is \emph{Pareto optimal} (\PO{}) if there is no other allocation that Pareto dominates it. 
We also consider a relaxed efficiency notion of {\em maximality}. An allocation $\mathcal{A}$ is \emph{maximal} if for every agent $i \in N$ and every unallocated item $o \in M \setminus \bigcup_{i \in N}A_i$, $o$ is adjacent to some item in $A_i$. 


\paragraph{Connection between goods and chores.}
For identical valuations, it is easy to see that the existence of a maximal EF1 allocation is equivalent for goods and chores.
Specifically, an allocation $\mathcal{A}$ is envy-free up to one good under valuation $v$ if and only if $\mathcal{A}$ is envy-free up to one chore under valuation $-v$. Thus, a maximal (respectively,~complete) allocation that is envy-free up to one good exists for $v$ if and only if a maximal (respectively,~complete) allocation that is envy-free up to one chore exists for $-v$.

\begin{restatable}
{proposition}{Equivalence}
Given a conflict graph $G$, let $\I^g$ denote an instance with $n$ agents with identical monotone non-decreasing valuation $v$, and let $\I^c$  denote an instance with $n$ agents with identical monotone non-increasing valuation $-v$. Then, an allocation $\A$ is envy-free up to one good $(\EF{1} \text{ for goods})$ in $\I^g$ if and only if $\A$ is envy-free up to one chore $(\EF{1} \text{ for chores})$ in $\I^c$. Similarly, an allocation $\A$ is envy-free up to any good $(\EFX{} \text{ for goods})$ in $\I^g$ if and only if $\A$ is envy-free up to any chore $(\EFX{} \text{ for chores})$ in $\I^c$.
\label{thm:Equivalence}
\end{restatable}

\section{Non-Existence Results}
\label{sec:Non-Existence-Results}

In this section, we will show that various combinations of fairness and efficiency notions can fail to exist in our problem. Interestingly, all of our counterexamples involve two agents with identical additive valuations. Further, with the exception of \Cref{eg:EF1_PO_goods_cycle_counterexample}, the conflict graph is a path graph. 
Let us start with a negative result for \EFX{} and maximality.


\begin{example}[Non-existence of \EFX{}+maximality]
Consider an instance with two agents $1$ and $2$ and four goods $o_1$, $o_2$, $o_3$, $o_4$ that are identically valued by the agents at $1$, $1$, $1$, and $4$, respectively. The conflict graph is a path, as shown below.

\begin{figure}[h]
\centering
\begin{tikzpicture}
    \tikzset{mynode/.style = {shape=circle,draw,inner sep=1pt}} 
    \node[mynode] (1) at (0,0) {$o_1$};
    \node[mynode] (2) at (1,0) {$o_2$};
    \node[mynode] (3) at (2,0) {$o_3$};
    \node[mynode] (4) at (3,0) {$o_4$};
    \draw (1) -- (2);
    \draw (2) -- (3);
    \draw (3) -- (4);
\end{tikzpicture}
\label{fig:EFX_maximal_counterexample}
\end{figure}
Let $\A$ be a maximal allocation. By maximality, the good $o_4$ must be allocated: if $o_3$ is allocated to one agent, then $o_4$ can be allocated to the other agent, while if $o_3$ is unallocated, then $o_4$ can be allocated to either agent. Thus, without loss of generality, suppose that $o_4$ is assigned to agent $1$.

We claim that, if $\A$ satisfies \EFX{}, then agent $1$ cannot receive any other good. Indeed, agent $2$ can receive at most two of the remaining goods $o_1,o_2,o_3$, and hence $v(A_2)\leq 2$. If agent $1$ receives $o_4$ together with some other good $o$, then removing $o$ from $A_1$ leaves value at least $4$, which is still larger than $v(A_2)$. Hence agent $2$ envies agent $1$ even after the removal of the good $o$, contradicting \EFX{}.

It remains to show that maximality is then impossible. If agent $2$ receives exactly two of the goods $o_1,o_2,o_3$, feasibility forces $A_2=\{o_1,o_3\}$; but then the unallocated good $o_2$ can be assigned to agent $1$, contradicting maximality. If agent $2$ receives at most one of $o_1,o_2,o_3$, then at least one of $o_1$ or $o_2$ is unallocated and can be feasibly assigned to agent $1$, again contradicting maximality. Therefore, no \EFX{} and maximal allocation exists.\qed
\label{eg:EFX_Maximal_Counterexample}
\end{example}

The non-existence of \EFX{} motivates the consideration of a weaker approximation such as \EF{1}. Our next example shows that \EF{1} and Pareto optimality can be mutually incompatible even for two agents with identical additive goods valuations on a cycle graph.

\begin{figure}[t]
	\begin{center}
	    \scalebox{1}{
	       \begin{tikzpicture}[scale=0.9, every node/.style={scale=0.9}]
                \centering
	            \tikzstyle{onlytext}=[text width = 3cm, align = center]
	            \tikzset{venn circle/.style={circle,minimum width=0mm,fill=#1,opacity=0.1}}
                    \node[onlytext] (Result) at (7.45,-0.75) {{Results}};
                    %
                    %
                    \draw [line width=28pt,opacity=0.4,color=red!55,line cap=round,rounded corners] (-0.5,-2) -- (0.5,-2) -- (2.5,-4.5) -- (3.5,-4.5) -- (5,-2) -- (9.5,-2);
                    \draw [line width=28pt,opacity=0.4,color=red!55,line cap=round,rounded corners] (-0.5,-4.5) -- (0.5,-4.5) -- (2.5,-2) -- (3.5,-2) -- (5,-3.25) -- (9.5,-3.25);
                    \draw [line width=28pt,opacity=0.4,color=green!55,line cap=round,rounded corners] (-0.5,-4.5) -- (9.5,-4.5);
	            \node[onlytext] () at (7.45,-2) {May not exist\\(\Cref{eg:EFX_Maximal_Counterexample})};
                    %
                    %
	            \node[onlytext] () at (7.45,-3.25) {May not exist\\(\Cref{eg:EF1_PO_line_graph_counterexample})};
                    %
                    %
	            \node[onlytext] () at (7.2,-4.5) {\begin{tabular}{c}{Exists in certain settings}\\{(See \Cref{tab:Summary})}\end{tabular}};
                    \draw[-, line width=1pt] (-1,-1.25) -- (10,-1.25);
                    \node[onlytext] (Efficiency) at (3,-0.75) {\begin{tabular}{c}{Efficiency}\\{Notions}\end{tabular}};
                    \node[onlytext] (PO) at (3,-2) {\PO{}};
                    \node[onlytext] (Maximal) at (3,-4.5) {Maximal};
                    \node[onlytext] (Fairness) at (0,-0.75) {\begin{tabular}{c}{Fairnesss}\\{Notions}\end{tabular}};
                    \node[onlytext] (EFX) at (0,-2) {\EFX{}};
                    \node[onlytext] (EFone) at (0,-4.5) {\EF{1}};
                    \draw[->, line width=1pt] (EFX) -- (EFone);
                    \draw[->, line width=1pt] (PO) -- (Maximal);
	       \end{tikzpicture}
            }
	\end{center}
	\caption{Summary of our results. The arrows denote logical implications between fairness and efficiency notions. The positive and negative results are shown in green and red, respectively.}
	\label{fig:relations}
\end{figure}
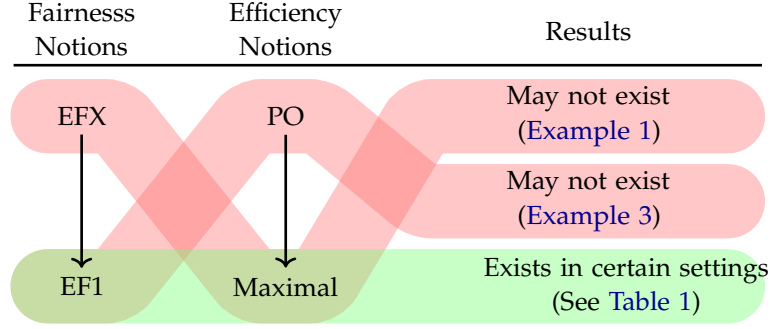

\begin{example}[Non-existence of  \EF{1}+Pareto optimality for goods]
Consider an instance with two agents $1$ and $2$ and six goods $o_1,\dots,o_6$ that are identically valued by the agents at $1,2,1,2,1,2$, respectively. The conflict graph is a cycle, as shown below.

\begin{figure}[h]
\centering
\begin{tikzpicture}
    \tikzset{mynode/.style = {shape=circle,draw,inner sep=1pt}} 
    \node[mynode] (1) at (0:1.4) {$o_1$};
    \node[mynode] (2) at (60:1.4) {$o_2$};
    \node[mynode] (3) at (120:1.4) {$o_3$};
    \node[mynode] (4) at (180:1.4) {$o_4$};
    \node[mynode] (5) at (240:1.4) {$o_5$};
    \node[mynode] (6) at (300:1.4) {$o_6$};
    \draw (1) -- (2);
    \draw (2) -- (3);
    \draw (3) -- (4);
    \draw (4) -- (5);
    \draw (5) -- (6);
    \draw (6) -- (1);
\end{tikzpicture}
\label{fig:EF1_PO_goods_cycle_counterexample}
\end{figure}

Let $L=\{o_1,o_3,o_5\}$ be the set of low-value goods and $H=\{o_2,o_4,o_6\}$ be the set of high-value goods. Since both $L$ and $H$ are independent sets, the allocations $(L,H)$ and $(H,L)$ are feasible and have value profiles $(3,6)$ and $(6,3)$, respectively.

We first characterize the Pareto optimal allocations. Any agent whose bundle does not contain all three high-value goods has value at most $4$: with no high-value good its value is at most $3$, with one high-value good its value is at most $3$, and with two high-value goods it cannot receive any low-value good. Moreover, at most one agent can have value $4$, since value $4$ requires receiving two high-value goods. Hence, if neither agent receives all goods in $H$, then one of the allocations $(L,H)$ or $(H,L)$ Pareto dominates the allocation: if one agent has value $4$, give $H$ to that agent and $L$ to the other; otherwise, both agents have value at most $3$. If one agent receives all goods in $H$, then the other agent must receive all goods in $L$ in any Pareto optimal allocation; otherwise, assigning the missing low-value goods to the other agent gives a Pareto improvement. Therefore, every Pareto optimal allocation is either $(L,H)$ or $(H,L)$.

Neither of these two allocations satisfies \EF{1}. In each case, the agent receiving $L$ has value $3$, while the agent receiving $H$ has value $6$; even after removing any one high-value good from $H$, the remaining value is $4>3$. Thus, no allocation satisfies both \EF{1} and Pareto optimality.\qed
\label{eg:EF1_PO_goods_cycle_counterexample}
\end{example}

We present a similar example for chores on a path graph.

\begin{example}[Non-existence of  \EF{1}+Pareto optimality for chores]
Consider an instance with two agents $1$ and $2$ and five items $o_1, \dots, o_5$ that are identically valued by the agents at $-2$, $-10$, $-1$, $-10$, and $-2$, respectively.  The conflict graph is a path graph, as shown below:

\begin{figure}[h]
\centering
\begin{tikzpicture}
    \tikzset{mynode/.style = {shape=circle,draw,inner sep=1pt}} 
    \node[mynode] (1) at (0,0) {$o_1$};
    \node[mynode] (2) at (1,0) {$o_2$};
    \node[mynode] (3) at (2,0) {$o_3$};
    \node[mynode] (4) at (3,0) {$o_4$};
    \node[mynode] (5) at (4,0) {$o_5$};
    \draw (1) -- (2);
    \draw (2) -- (3);
    \draw (3) -- (4);
    \draw (4) -- (5);
\end{tikzpicture}
\label{fig:EF1_PO_line_graph_counterexample}
\end{figure}

Any maximal allocation that allocates a ``heavy'' item ($o_2$ or $o_4$) to one of the agents, say agent $1$, can be shown to be Pareto dominated by an allocation that assigns the extreme items ($o_1$ and $o_5$) to agent $1$ and the middle item $o_3$ to agent $2$. Therefore, any maximal allocation that leaves both heavy items unassigned must be of the form $(\{o_1,o_5\},\{o_3\})$ or $(\{o_3\},\{o_1,o_5\})$, neither of which satisfy \EF{1}.\qed
\label{eg:EF1_PO_line_graph_counterexample}
\end{example}


Note that a Pareto optimal allocation (without \EF{1}) always exists; in particular, a allocation that maximizes the sum of agents' utilities~(i.e., the utilitarian social welfare maximizing allocation) over the space of all maximal allocations is Pareto optimal.

Another notion of efficiency is \emph{completeness} which asks that no chore should be left unassigned. A complete allocation may not exist, such as with two agents and a triangle conflict graph. However, even when a complete allocation exists, no such allocation may satisfy \EF{1}.

\begin{example}[Non-existence of \EF{1}+completeness]
Consider an instance with two agents $1$ and $2$ and four goods $o_1,o_2,o_3,o_4$ such that the odd-indexed goods are valued at $1$ each and the even-indexed goods are valued at $3$ each by both agents. The conflict graph is a path graph, as in \Cref{eg:EFX_Maximal_Counterexample}.

Any complete allocation must assign the two odd-indexed goods to one agent and the two even-indexed goods to the other agent. The former agent has value $2$, while the latter agent has value $6$. Even after removing either even-indexed good from the latter bundle, its value remains $3>2$. Hence every complete allocation violates \EF{1}.\qed
\label{eg:EF1_complete_counterexample}
\end{example}




\subsection{Limitations of Algorithms from the Unconstrained Setting}
\label{subsec:Limitations}

Part of what makes the fair division with conflict constraints problem challenging is that algorithmic techniques from unconstrained fair division do not automatically extend to the constrained problem. We will illustrate this challenge through two examples that demonstrate that the well-known round robin and envy-cycle elimination algorithms, which satisfy \EF{1} for unconstrained items, fail to do so in the presence of conflicts.

The round robin algorithm iterates over the agents according to a fixed permutation, and each agent, on its turn, picks its favorite remaining item. This algorithm is known to find an \EF{1} allocation under additive valuations in the unconstrained problem. However, when the conflict graph is a path graph, round robin can fail to achieve \EF{1} even for two agents with identical valuations.

\begin{example}[Round robin fails \EF{1}]
Consider an instance with eight items $o_1,\dots,o_8$ and two agents $1$ and $2$ with identical additive valuations. The valuation profile and the conflict graph are shown below.
%
\begin{figure}[ht]
\centering
	\begin{tikzpicture}[scale=0.8]
		\tikzset{mynode/.style = {shape=circle,draw,inner sep=1pt}} 
            \node[] (0) at (0,2.7) {};
		\node[mynode] (1) at (0,2) {$o_1$};
		\node[mynode] (2) at (1,2) {$o_2$};
		\node[mynode] (3) at (2,2) {$o_3$};
		\node[mynode] (4) at (3,2) {$o_4$};
		\node[mynode] (5) at (4,2) {$o_5$};
		\node[mynode] (6) at (5,2) {$o_6$};
		\node[mynode] (7) at (6,2) {$o_7$};
		\node[mynode] (8) at (7,2) {$o_8$};
		\draw (1) -- (2);
		\draw (2) -- (3);
		\draw (3) -- (4);
		\draw (4) -- (5);
		\draw (5) -- (6);
		\draw (6) -- (7);
		\draw (7) -- (8);
		\node[] (1) at (-0.75,1.25) {1:};
		\node[] (1) at (0,1.25) {10};
		\node[] (2) at (1,1.25) {3};
		\node[] (3) at (2,1.25) {8};
		\node[] (4) at (3,1.25) {9};
		\node[] (5) at (4,1.25) {7};
		\node[] (6) at (5,1.25) {2};
		\node[] (7) at (6,1.25) {1};
		\node[] (8) at (7,1.25) {0};
		\node[] (1) at (-0.75,0.5) {2:};
		\node[] (1) at (0,0.5) {10};
		\node[] (2) at (1,0.5) {3};
		\node[] (3) at (2,0.5) {8};
		\node[] (4) at (3,0.5) {9};
		\node[] (5) at (4,0.5) {7};
		\node[] (6) at (5,0.5) {2};
		\node[] (7) at (6,0.5) {1};
		\node[] (8) at (7,0.5) {0};
	\end{tikzpicture}
\end{figure}

Suppose agent $1$ goes first in the round robin algorithm. Then the resulting allocation is $A_1 = \{o_1,o_3,o_5,o_7\}$ and $A_2 = \{o_4,o_2,o_6,o_8\}$. Note that $v_2(A_2) < v_2(A_1\setminus \{o\})$ for every $o \in A_1$, implying that the allocation fails \EF{1}. 

The reason why round robin fails \EF{1} in this example is that after picking $o_4$ in the first round, agent $2$ can no longer pick $o_5$ in its next turn due to feasibility constraint. Thus, it has to pick its favorite among the feasible items, which is $o_2$. This results in a ``build up'' of envy in each round, which cannot be compensated even after removing the worst item $o_8$ in its bundle.\qed
\end{example}

The envy-cycle elimination algorithm is known to find an \EF{1} allocation for indivisible goods in the unconstrained setting~\citep{LMM+04approximately}. 
The algorithm works as follows: At each step, the algorithm assigns an unallocated good to a ``source'' vertex in the envy graph. (The envy graph associated with a given allocation is a directed graph whose vertices are the agents and there is an edge $(i,j)$ if agent $i$ envies agent $j$ in the given allocation. A source vertex refers to an agent who is not envied by any other agent.) If a source does not exist, there must exist a cycle of envy edges. By repeatedly resolving such cycles (i.e., by cyclically swapping bundles), one can eventually make one of the vertices a source.

\begin{example}[Envy-cycle elimination fails \EF{1}]
Consider an instance with two agents $1$ and $2$ with identical additive valuations and five items $o_1,\dots,o_5$ whose conflict graph is a path graph (similar to \Cref{eg:EFX_Maximal_Counterexample}). The items are valued at $4$, $0$, $2$, $3$, and $2$, respectively, by both agents. Note that since the valuations are identical, an envy cycle can never occur. Thus, a source agent always exists.

At each step of the algorithm, we allow a source agent to choose its favorite feasible item. Thus, agent $1$ starts by picking the item $o_1$, which makes agent $2$ the source. Next, agent $2$ picks the item $o_4$, which again makes agent $2$ the source. However, on its next turn, agent $2$ cannot pick either $o_3$ or $o_5$ due to conflict constraints. Hence, it must settle for item $o_2$. It can be checked that the induced allocation is $(\{o_1,o_3,o_5\},\{o_2,o_4\})$. However, \EF{1} is violated from agent $2$'s perspective.\qed
\end{example}

\section{Results for Two Agents}
\label{sec:Results_Two_Agents}

In this section, we will present our results for the case of two agents. Our main result is that under monotone valuations, an \EF{1} and maximal allocation always exists for any conflict graph. Furthermore, when the valuations take integer values, such an allocation can be computed in pseudopolynomial time~(\Cref{thm:TwoAgents-GeneralGraph-MonotoneVals-Existence}). 
In \cref{sec:two-agents:specialized}, we will later consider specialized settings where either the conflict graph or the valuation functions have more structure. For these settings, we will show that an \EF{1} allocation can be computed in polynomial time without assuming integrality of valuations.

\begin{restatable}{theorem}{TwoAgents-GeneralGraph-MonotoneVals-Existence}
For two agents with monotone valuation functions and any conflict graph, there exists an allocation that is both \EF{1} and  maximal. Moreover, for integral valuations, such an allocation can be computed in pseudopolynomial time.
\label{thm:TwoAgents-GeneralGraph-MonotoneVals-Existence}
\end{restatable}

The first step in the proof of \cref{thm:TwoAgents-GeneralGraph-MonotoneVals-Existence} is a reduction to the case where both the agents have \emph{identical} valuations and all the items are goods. Formally, we show that \cref{thm:two-agents:monotone-goods} below implies \cref{thm:TwoAgents-GeneralGraph-MonotoneVals-Existence}. This part of the proof is a fairly standard argument and is presented next. Note that the reduction to the identical valuations setting works even when the valuations are not monotone. Later, we prove \cref{thm:two-agents:monotone-goods}.

\begin{restatable}{theorem}{two-agents:monotone-goods}
\label{thm:two-agents:monotone-goods}
For two agents with identical monotone non-decreasing valuation functions and any conflict graph,  there exists an allocation that is both \EF{1} and maximal. Moreover, for integral valuations, such an allocation can be computed in pseudopolynomial time.
\end{restatable}


To show that \cref{thm:two-agents:monotone-goods} implies \cref{thm:TwoAgents-GeneralGraph-MonotoneVals-Existence}, we will make use of~\Cref{lemma:two-agents:monotone-goods}. This result employs a standard ``cut-and-choose'' argument, which allows us to assume, without loss of generality, that both agents have identical valuations. Note that this lemma does not require the valuations to be monotone.

\begin{lemma}
Consider $n = 2$ agents and any set of items. 
Let $v$ and $v'$ be (not necessarily monotone) valuation functions. For any conflict graph, if $\mathcal{A} = \paren{ A_1, A_2 }$ is a maximal and \EF{1} allocation when both agents have valuation $v$, then either $\mathcal{A}$ or $\mathcal{A}^{ \mathsf{rev} } = \paren{ A_2, A_1 }$ is maximal and \EF{1} when one of the agents has valuation $v$ and the other has valuation $v'$.
\label{lemma:two-agents:monotone-goods}
\end{lemma}
\begin{proof}
Note that maximality is straightforward and we only need to show the \EF{1} property. For this, assume without loss of generality that the first agent has valuation $v$ and the second agent has valuation $v'$. With this assumption, note that both the allocations $\mathcal{A}$ and $\mathcal{A}^{ \mathsf{rev} }$ are \EF{1} from the first agent's perspective. As one of the them has to be envy-free (and therefore \EF{1}) from the second agent's perspective, we are done.
\end{proof}



We can now show that \cref{thm:two-agents:monotone-goods} implies \cref{thm:TwoAgents-GeneralGraph-MonotoneVals-Existence}.

\begin{proof}[Proof of \Cref{thm:TwoAgents-GeneralGraph-MonotoneVals-Existence} assuming \Cref{thm:two-agents:monotone-goods}]
Let $v_1$ and $v_2$ be the monotone valuation functions of the two agents, and let $v$ be the valuation function defined as $v(S) = \abs{ v_1(S) }$ for all sets $S$ of items. Observe that $v$ is monotone non-decreasing. From \cref{thm:two-agents:monotone-goods}, we get an allocation $\mathcal{A} = \paren{ A_1, A_2 }$ that is both maximal and \EF{1} when both agents have valuations $v$ in pseudopolynomial time. Furthermore, any value query for the valuation function $v$ can be simulated using a value query for the valuation function $v_1$ by taking the absolute value of the oracle's output.

Observe that $v_1 = v$ if $v_1$ is non-decreasing, and $v_1 = -v$ if $v_1$ is non-increasing. Thus, from \Cref{thm:Equivalence}, we get that $\mathcal{A}$ is both maximal and \EF{1} when both agents have valuations $v_1$. Finally, from \Cref{lemma:two-agents:monotone-goods}, we get that either $\mathcal{A}$ or $\mathcal{A}^{ \mathsf{rev} } = \paren{ A_2, A_1 }$ is maximal and \EF{1} when the agents have valuations $v_1$ and $v_2$, respectively.
\end{proof}

\subsection{Proof for Two Identical Agents and Monotone Goods}


We now focus our attention on proving \cref{thm:two-agents:monotone-goods}. Throughout, we use $v$ to denote the identical monotone non-decreasing valuation function of the agents. Our key lemma~(\Cref{lemma:EF1-from-order-adjacent}) shows that the existential guarantee of~\Cref{thm:two-agents:monotone-goods} follows if we can construct a sequence of `order-adjacent' allocations, as defined next.

\begin{definition}
\label{def:order-adjacent}
We say that an ordered pair of allocations $\paren{ \mathcal{A} = \paren{ A_1, A_2 }, \mathcal{A}' = \paren{ A'_1, A'_2 } }$ is \emph{order-adjacent} if we have $\card{ A_1 \setminus A'_1 } \leq 1$ and $\card{ A'_2 \setminus A_2 } \leq 1$. 
\end{definition}

We emphasize that \cref{def:order-adjacent} is not symmetric. Namely, it is possible for $\paren{ \mathcal{A}, \mathcal{A}' }$ to be order-adjacent without $\paren{ \mathcal{A}', \mathcal{A} }$ being so. 

The following lemma is the key consequence of order-adjacency that we use in our proof. Roughly speaking, it shows that if agent 1 is the envied agent in $\mathcal{A}$ and agent 2 is envied in $\mathcal{A}'$, 
then at least one of the allocations in the pair must be \EF{1}. Note that, because of identical valuations, at most one agent can be envied under any allocation.

\begin{lemma}
Let $\paren{ \mathcal{A} = \paren{ A_1, A_2 }, \mathcal{A}' = \paren{ A'_1, A'_2 } }$ be an order-adjacent pair of allocations. If we have that $v\paren{ A_1 } \geq v\paren{ A_2 }$ and $v\paren{ A'_2 } \geq v\paren{ A'_1 }$, then at least one of the allocations $\mathcal{A}$ and $\mathcal{A}'$ is \EF{1}.\footnote{Note that just like \cref{def:order-adjacent}, these two inequalities are not symmetric.}
\label{lemma:EF1-from-order-adjacent}
\end{lemma}
%
\begin{proof}
Note the following inequality:
\begin{align*}
v\paren{ A_2 } + v\paren{ A'_1 } &\geq v\paren{ A_2 \cap A'_2 } + v\paren{ A_1 \cap A'_1 } \tag{$v$ is a monotone valuation for goods} \\
&= v\paren{ A'_2 \setminus \paren{ A'_2 \setminus A_2 } } + v\paren{ A_1 \setminus \paren{ A_1 \setminus A'_1 } } .
\end{align*}
Thus, we either have $v\paren{ A_2 } \geq v\paren{ A_1 \setminus \paren{ A_1 \setminus A'_1 } }$ or we have $v\paren{ A'_1 } \geq v\paren{ A'_2 \setminus \paren{ A'_2 \setminus A_2 } }$. As both $\card{ A_1 \setminus A'_1 }, \card{ A'_2 \setminus A_2 } \leq 1$ from \cref{def:order-adjacent}, we get that either $\mathcal{A}$ or $\mathcal{A}'$ is \EF{1}, as desired. 
\end{proof}

Henceforth, we focus on constructing an order-adjacent pair of allocations satisfying \cref{lemma:EF1-from-order-adjacent}. For this, we first describe \cref{algo:sequence-from-good-S} that starts with any maximal independent set $S$ in the conflict graph and uses $S$ to construct a sequence of maximal allocations $\paren{ \mathcal{A}_o }_{ o \in \set{ 0 } \cup [m] }$. We prove that if the set $S$ satisfies a certain property, then at least one of the allocations in the sequence is \EF{1}. The exact property that we need is described later in \cref{lemma:spropforproof}, but we note here that the maximal independent set in the conflict graph $S$ that maximizes $v\paren{ S }$ among all maximal independent sets trivially satisfies this property. Thus, by starting with this set $S$, we get \cref{thm:two-agents:monotone-goods}. Computing the set $S$ may not be efficient in general, so we make no guarantees about the running time. Later on, we show that in some specialized settings, an appropriate $S$ can be found efficiently. We use the notation $\overline{ S }$ to denote $M \setminus S$. 

In the interest of notational simplicity, we will index the goods as $1,2,\dots,m$ instead of $o_1,o_2,\dots,o_m$ for the rest of this section. Additionally, for any $r \in \{1,2,\dots,m\}$, we will write $[r]$ to denote $\{1,2,\dots,r\}$.


\begin{algorithm}[H]
\caption{Constructing a sequence of order-adjacent allocations. }
\label{algo:sequence-from-good-S}
\begin{algorithmic}[1]

\renewcommand{\algorithmicrequire}{\textbf{Input:}}
\renewcommand{\algorithmicensure}{\textbf{Output:}}

\smallskip

\Require A maximal independent set $S$ in the conflict graph.
\Ensure A sequence $\paren{ \mathcal{A}_o }_{ o \in \set{ 0 } \cup [m] }$ of order-adjacent allocations.

\smallskip

\State For all $o \in \overline{ S }$, we define: \label{line:max-min-conf} \Comment{{\footnotesize\emph{Note that the set $\set{ o' \in S \mid \paren{ o, o' } \in E }$ cannot be empty due to maximality of $S$.}}}
\begin{align*}
\maxconf\paren{ o } = \max \set{ o' \in S \mid \paren{ o, o' } \in E } , \\
\minconf\paren{ o } = \min \set{ o' \in S \mid \paren{ o, o' } \in E } . 
\end{align*} 

\smallskip

\State Sort all $o \in \overline{ S }$ in increasing order of $\maxconf\paren{ o }$ (breaking ties arbitrarily). In this order, greedily pick vertices to get a maximal independent set $X_{ \maxconf }$. Define $X_{ \minconf }$ analogously with the vertices sorted in decreasing order of $\minconf\paren{ o }$. \label{line:x-max-min-conf}

\smallskip

\State For all $o \in \set{ 0 } \cup [m]$, define the allocation $\mathcal{A}_o = \paren{ A_{ o, 1 }, A_{ o, 2 } }$ where: \label{line:ao} \Comment{{\footnotesize\emph{Recall that $[o] = \{1,2,\dots,o\}$.}}}
\begin{align*}
A_{ o, 1 } &= \paren{ S \setminus [o] } \cup \set{ o' \in X_{ \maxconf } \mid \maxconf\paren{ o' } \leq o } , \\
A_{ o, 2 } &= \paren{ S \cap [o] } \cup \set{ o' \in X_{ \minconf } \mid \minconf\paren{ o' } > o } .
\end{align*}

\alglinenoNew{algcommon}
\alglinenoPush{algcommon}

\end{algorithmic}
\end{algorithm}

In order to give some intuition about \cref{algo:sequence-from-good-S}, we illustrate its behavior on two simple graphs in \cref{fig:algo-trace}, namely the path graph and a simple example of a chordal graph.\footnote{A chordal graph is an undirected graph where every cycle of four or more vertices has a \emph{chord}, i.e., an edge connecting two non-adjacent vertices of the cycle.} Both graphs have eight vertices, and the vertices in the set $S$ input to the algorithm are marked as shaded. It can be verified that the shaded vertices form a maximal independent set in the respective graph. For some vertices, the values $\maxconf\paren{ \cdot }$ and $\minconf\paren{ \cdot }$ computed in \cref{line:max-min-conf} are depicted using curved arrows. Finally, the assignments $\paren{ \mathcal{A}_o = \paren{ A_{ o, 1 }, A_{ o, 2 } } }_{ o \in \set{ 0 } \cup [m] }$ are shown in order under the graphs, with the elements of $A_{ o, 1 }$ colored red and the elements of $A_{ o, 2 }$ colored blue. As $A_0 = \paren{ S, X_{ \minconf } }$ and $A_m = \paren{ X_{ \maxconf }, S }$, this also depicts the values of $X_{ \maxconf }$ and $X_{ \minconf }$. Note that the owner of the set $S$ reverses from red to blue as the allocation changes from $\mathcal{A}_0$ to $\mathcal{A}_m$. Consequently, whenever $v\paren{ S } \geq \max \set{ v\paren{ X_{ \minconf } }, v\paren{ X_{ \maxconf } } }$, the envy relation must switch at some point along the sequence.  


\newcommand{\mygraph}[3]{%
\begin{scope}[yshift=#2]
    \foreach \col [count=\i] in {#3} {
        \node[vertex, fill=\col] (w#1\i) at (0.8 * \i - 3.6,0) {\i};
    }
    \foreach \i in {1,...,7} {
        \pgfmathtruncatemacro{\j}{\i+1}
        \draw (w#1\i) -- (w#1\j);
    }
\end{scope}
}

\newcommand{\mygraphB}[3]{%
\begin{scope}[yshift=#2]

\def\colors{#3}

\foreach \col [count=\i] in \colors {

    \ifnum\i=1 \node[vertex,fill=\col] (w#1\i) at (-1.75,0) {1};\fi
    \ifnum\i=2 \node[vertex,fill=\col] (w#1\i) at (-1.75,0.75) {2};\fi
    \ifnum\i=3 \node[vertex,fill=\col] (w#1\i) at (0,0.75) {3};\fi
    \ifnum\i=4 \node[vertex,fill=\col] (w#1\i) at (1.75,0.75) {4};\fi
    \ifnum\i=5 \node[vertex,fill=\col] (w#1\i) at (1.75,0) {5};\fi
    \ifnum\i=6 \node[vertex,fill=\col] (w#1\i) at (1.75,-0.75) {6};\fi
    \ifnum\i=7 \node[vertex,fill=\col] (w#1\i) at (0,-0.75) {7};\fi
    \ifnum\i=8 \node[vertex,fill=\col] (w#1\i) at (-1.75,-0.75) {8};\fi

}

\draw (w#11)--(w#12)--(w#13)--(w#14)--(w#15)--(w#16)--(w#17)--(w#18)--(w#11);
\draw (w#13)--(w#15)--(w#17)--(w#11)--(w#13);
\draw (w#13)--(w#17);

\end{scope}
}

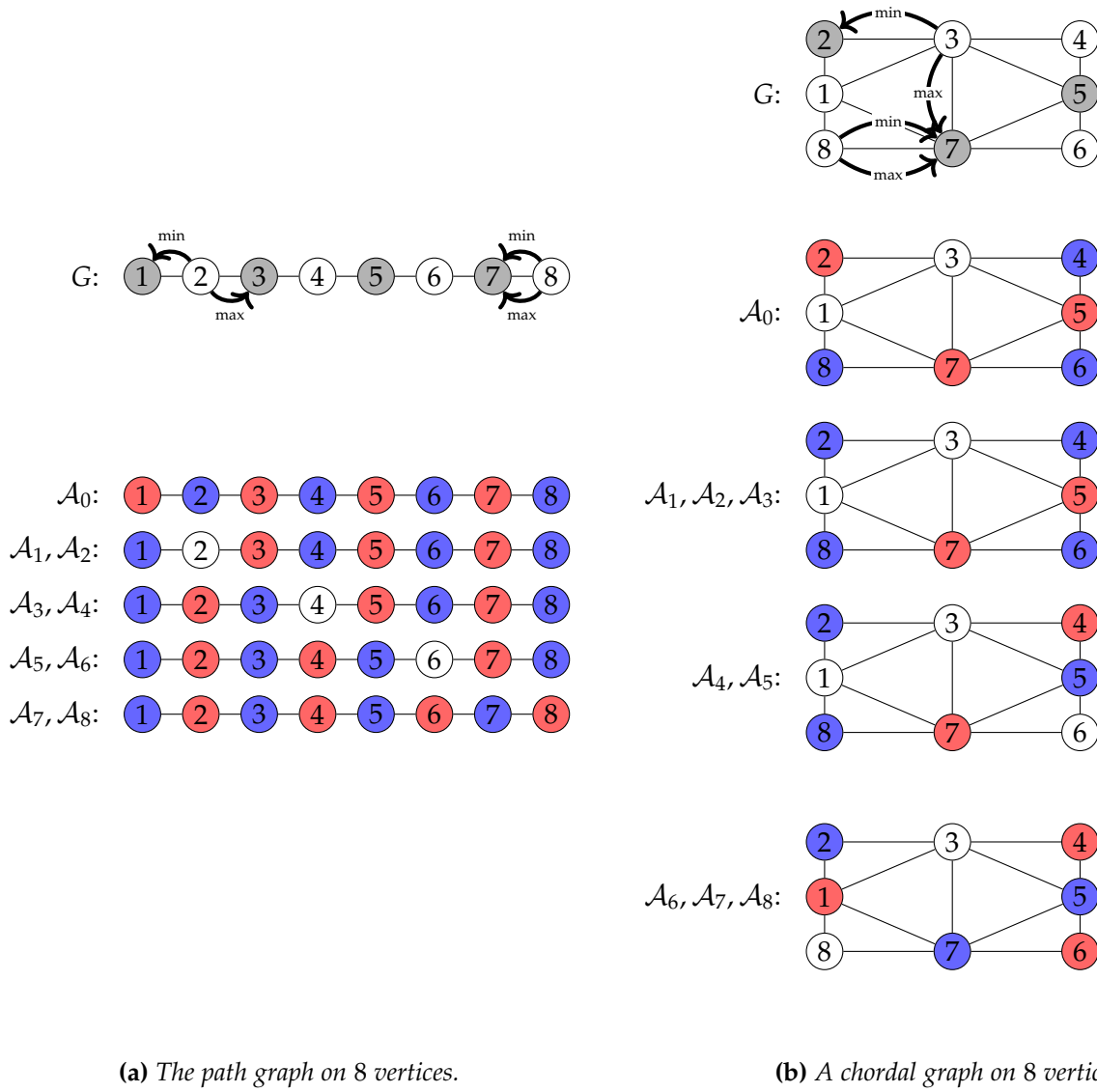
\begin{figure}[t]
\centering

\begin{subfigure}[t]{0.45\textwidth}
\centering
\begin{tikzpicture}[
    vertex/.style={
        circle,
        draw,
        minimum size=5mm,
        inner sep=0pt
    }
]

\useasboundingbox (-4.5,-7.5) rectangle (4.5,7.5);

\mygraph{1}{3cm}{gray!60, white, gray!60, white, gray!60, white, gray!60, white}
\node[left=0.2cm of w11] {$G$:};

\draw[->, bend right=55, ultra thick] (w12) to node[above]{\tiny min} (w11);
\draw[->, bend right=55, ultra thick] (w12) to node[below]{\tiny max} (w13);

\draw[->, bend right=55, ultra thick] (w18) to node[above]{\tiny min} (w17);
\draw[->, bend left=55, ultra thick] (w18) to node[below]{\tiny max} (w17);

\mygraph{2}{0cm}{red!60, blue!60, red!60, blue!60, red!60, blue!60, red!60, blue!60}
\node[left=0.2cm of w21] {$\mathcal{A}_0$:};

\mygraph{3}{-0.75cm}{blue!60, white, red!60, blue!60, red!60, blue!60, red!60, blue!60}
\node[left=0.2cm of w31] {$\mathcal{A}_1, \mathcal{A}_2$:};

\mygraph{4}{-1.5cm}{blue!60, red!60, blue!60, white, red!60, blue!60, red!60, blue!60}
\node[left=0.2cm of w41] {$\mathcal{A}_3, \mathcal{A}_4$:};

\mygraph{5}{-2.25cm}{blue!60, red!60, blue!60, red!60, blue!60, white, red!60, blue!60}
\node[left=0.2cm of w51] {$\mathcal{A}_5, \mathcal{A}_6$:};

\mygraph{6}{-3cm}{blue!60, red!60, blue!60, red!60, blue!60, red!60, blue!60, red!60}
\node[left=0.2cm of w61] {$\mathcal{A}_7, \mathcal{A}_8$:};

\end{tikzpicture}
\caption{The path graph on $8$ vertices.}
\label{fig:left}
\end{subfigure}
\hfill
\begin{subfigure}[t]{0.45\textwidth}
\centering
\begin{tikzpicture}[
    vertex/.style={
        circle,
        draw,
        minimum size=5mm,
        inner sep=0pt
    },
    middle/.style={midway, inner sep = 1pt, fill=white}
]

\useasboundingbox (-3,-7.5) rectangle (3,7.5);

\mygraphB{1}{5.5cm}{white, gray!60, white, white, gray!60, white, gray!60, white}
\node[left=0.2cm of w11] {$G$:};

\draw[->, bend right=35, ultra thick] (w13) to node[middle]{\tiny max} (w17);
\draw[->, bend right=35, ultra thick] (w13) to node[middle]{\tiny min} (w12);
\draw[->, bend right=35, ultra thick] (w18) to node[middle]{\tiny max} (w17);
\draw[->, bend left=35, ultra thick] (w18) to node[middle]{\tiny min} (w17);

\mygraphB{2}{2.5cm}{white, red!60, white, blue!60, red!60, blue!60, red!60, blue!60}
\node[left=0.2cm of w21] {$\mathcal{A}_0$:};

\mygraphB{3}{0cm}{white, blue!60, white, blue!60, red!60, blue!60, red!60, blue!60}
\node[left=0.2cm of w31] {$\mathcal{A}_1, \mathcal{A}_2, \mathcal{A}_3$:};

\mygraphB{4}{-2.5cm}{white, blue!60, white, red!60, blue!60, white, red!60, blue!60}
\node[left=0.2cm of w41] {$\mathcal{A}_4, \mathcal{A}_5$:};

\mygraphB{5}{-5.5cm}{red!60, blue!60, white, red!60, blue!60, red!60, blue!60, white}
\node[left=0.2cm of w51] {$\mathcal{A}_6, \mathcal{A}_7, \mathcal{A}_8$:};

\end{tikzpicture}
\caption{A chordal graph on $8$ vertices.}
\label{fig:right}
\end{subfigure}

\caption{A depiction of \cref{algo:sequence-from-good-S} on two input graphs. It can be verified that all allocations are maximal and every pair of consecutive allocation is order-adjacent.}
\label{fig:algo-trace}

\end{figure}

\paragraph{Analyzing \cref{algo:sequence-from-good-S}.} We show that \cref{line:ao} defines maximal allocations.

\begin{lemma}
\label{lemma:ao-maximal}
For any input set $S$ and all $o \in \set{ 0 } \cup [m]$, the allocation $\mathcal{A}_o$ constructed in \cref{line:ao} is maximal.
\end{lemma}
\begin{proof}
We first prove that $\mathcal{A}_o$ is indeed an allocation by showing that $A_{ o, 1 } \cap A_{ o, 2 } = \emptyset$. For this, note from \cref{line:x-max-min-conf} that $S \cap X_{ \minconf } = S \cap X_{ \maxconf } = \emptyset$. Thus, it suffices to show that:
\[
\set{ o' \in \overline{ S } \mid \minconf\paren{ o' } > o } \cap \set{ o' \in \overline{ S } \mid \maxconf\paren{ o' } \leq o } = \emptyset .
\]
This is because \cref{line:max-min-conf} implies $\minconf\paren{ o' } \leq \maxconf\paren{ o' }$ for all $o' \in \overline{ S }$. 

We now show that the allocation $\mathcal{A}_o$ is maximal. For this, we observe the following from \cref{line:x-max-min-conf}.
\begin{observation}
\label{obs:x-max-min-conf}
For all $o' \in \overline{ S } \setminus X_{ \maxconf }$, there exists $o'' \in X_{ \maxconf }$ such that $\paren{ o', o'' } \in E$ and $\maxconf\paren{ o'' } \leq \maxconf\paren{ o' }$. Similarly, for all $o' \in \overline{ S } \setminus X_{ \minconf }$, there exists $o'' \in X_{ \minconf }$ such that $\paren{ o', o'' } \in E$ and $\minconf\paren{ o'' } \geq \minconf\paren{ o' }$.
\end{observation}
Now to show that the allocation $\mathcal{A}_o$ is maximal, note that for all $o' \notin A_{ o, 1 } \cup A_{ o, 2 }$, we must have that $o' \in \overline{ S }$. Recalling that $\minconf\paren{ o' } \leq \maxconf\paren{ o' }$ for all $o' \in \overline{ S }$, we get the following cases:
\begin{itemize}
\item \textbf{When $\minconf\paren{ o' } \leq \maxconf\paren{ o' } \leq o$:} In this case, we have $\minconf\paren{ o' }, \maxconf\paren{ o' } \in S \cap [o] \subseteq A_{ o, 2 }$. This means that $o' \notin A_{ o, 1 }$ implies that $o' \notin X_{ \maxconf }$ which by \cref{obs:x-max-min-conf} means that there exists $o'' \in X_{ \maxconf }$ such that $\paren{ o', o'' } \in G$ and $\maxconf\paren{ o'' } \leq \maxconf\paren{ o' }$. This means that $o'' \in A_{ o, 1 }$ and we are done.
\item \textbf{When $\minconf\paren{ o' } \leq o < \maxconf\paren{ o' }$:} In this case, we have $\minconf\paren{ o' } \in S \cap [o] \subseteq A_{ o, 2 }$ and $\maxconf\paren{ o' } \in S \setminus [o] \subseteq A_{ o, 1 }$. As they both conflict with $o'$, we are done.
\item \textbf{When $o < \minconf\paren{ o' } \leq \maxconf\paren{ o' }$:} In this case, we have $\minconf\paren{ o' }, \maxconf\paren{ o' } \in S \setminus [o] \subseteq A_{ o, 1 }$. This means that $o' \notin A_{ o, 2 }$ implies that $o' \notin X_{ \minconf }$ which by \cref{obs:x-max-min-conf} means that there exists $o'' \in X_{ \minconf }$ such that $\paren{ o', o'' } \in G$ and $\minconf\paren{ o'' } \geq \minconf\paren{ o' }$. This means that $o'' \in A_{ o, 2 }$ and we are done.\qedhere
\end{itemize}
\end{proof}

We now show that the constructed allocations are order-adjacent.

\begin{lemma}
\label{lemma:ao-order-adjacent}
For any input set $S$ and all $o \in [m]$, it holds that the pair of allocations $\paren{ \mathcal{A}_{ o - 1 }, \mathcal{A}_o }$ is order-adjacent.
\end{lemma}
\begin{proof}
By \cref{def:order-adjacent}, we have to show that $\card{ A_{ o - 1, 1 } \setminus A_{ o, 1 } } \leq 1$ and $\card{ A_{ o, 2 } \setminus A_{ o - 1, 2 } } \leq 1$. From \cref{line:ao}, this follows if we show that $\card{ S \cap \set{ o } } \leq 1$, which is direct.
\end{proof}


\begin{lemma}
\label{lemma:spropforproof}
For any input set $S$ that satisfies $v\paren{ S } \geq \max \set{ v\paren{ X_{ \minconf } }, v\paren{ X_{ \maxconf } } }$, there exists $o \in \set{ 0 } \cup [m]$ such that $\mathcal{A}_o$ is \EF{1}.
\end{lemma}
\begin{proof}
As $v\paren{ S } \geq \max \set{ v\paren{ X_{ \minconf } }, v\paren{ X_{ \maxconf } } }$, we have from \cref{line:ao} that $v\paren{ A_{ 0, 1 } } \geq v\paren{ A_{ 0, 2 } }$ and $v\paren{ A_{ m, 2 } } \geq v\paren{ A_{ m, 1 } }$. It follows that there exists $o \in [m]$ such that $v\paren{ A_{ o - 1, 1 } } \geq v\paren{ A_{ o - 1, 2 } }$ and $v\paren{ A_{ o, 2 } } \geq v\paren{ A_{ o, 1 } }$. Applying \cref{lemma:ao-order-adjacent,lemma:EF1-from-order-adjacent}, we get the theorem. 
\end{proof}

Observe that the existential guarantee in~\Cref{thm:two-agents:monotone-goods} directly follows from \cref{lemma:spropforproof} by taking $S$ to be the maximal independent set that maximizes $v\paren{ S }$. Indeed, with this choice, we have $v\paren{ S } \geq \max \set{ v\paren{ X_{ \minconf } }, v\paren{ X_{ \maxconf } } }$ and can apply \cref{lemma:spropforproof}.

We will now show that the allocation whose existence is guaranteed by \Cref{thm:two-agents:monotone-goods} can be computed in pseudopolynomial time. To prove this claim, we will show that the desired input $S$ for \Cref{algo:sequence-from-good-S}, i.e., a maximal independent set satisfying the condition outlined in \Cref{lemma:spropforproof}, can be computed in pseudopolynomial time. 

The procedure for computing the desired set $S$ is as follows: Start with an arbitrary maximal independent set $S$ in the graph $G$. Note that such a set can be computed via a greedy algorithm in polynomial time. Now, repeat the following steps until the condition in \Cref{lemma:spropforproof} is satisfied:
\begin{itemize}
    \item If the set $S$ satisfies the condition  in \Cref{lemma:spropforproof}, i.e., if $v\paren{ S } \geq \max \set{ v\paren{ X_{ \minconf } }, v\paren{ X_{ \maxconf } } }$, then return this set $S$.
    \item Otherwise, at least one of $v\paren{ X_{ \maxconf } } > v\paren{ S }$ or $v\paren{ X_{ \minconf } } > v\paren{ S }$ must be true. Let's assume the former condition is true; the argument for the latter is analogous. Starting from the set $X_{ \maxconf }$, which is an independent set of $G$, iteratively add vertices to it until we have a maximal independent set. Call this new set $S$ and repeat the first step.
\end{itemize}

The aforementioned procedure must terminate since the value $v(S)$ strictly increases in every iteration. Let $v_{\max} \coloneqq \max_{S \subseteq M} v(S)$ denote the maximum value of any set (recall that the valuations are monotone non-decreasing), and let $\Delta \coloneqq \min\{v(S) - v(T) : v(S) - v(T) > 0 \text{ for } S,T \subseteq M\}$ denote the smallest positive difference in the values of any two subsets. Thus, the value $v(S)$ increases by at least $\Delta$ in each iteration, and therefore the maximum number of iterations is at most $v_{\max} / \Delta$. For integral valuations, $\Delta \geq 1$ and therefore the number of iterations is at $v_{\max}$, implying a pseudopolynomial running time.

\subsection{Specialized Settings}
\label{sec:two-agents:specialized}

\cref{thm:TwoAgents-GeneralGraph-MonotoneVals-Existence} shows that a maximal and \EF{1} allocation exists for any conflict graph and any monotone valuation function but only provides a pseudopolynomial-time guarantee for computing such allocations. 
In this section, we describe some settings where such an allocation can be found efficiently.

\subsubsection{Bipartite Graphs}
\label{sec:two-agents:specialized:bipartite}

We start by considering the case where the conflict graph is bipartite.

\begin{restatable}{theorem}{TwoAgents-BipartiteGraph-MonotoneVals-Algorithm}
For two agents with monotone valuation functions and any bipartite conflict graph, a maximal and \EF{1} allocation exists and can be computed in polynomial time.
\label{thm:TwoAgents-BipartiteGraph-MonotoneVals-Algorithm}
\end{restatable}
\begin{proof}
As in the proof of \cref{thm:TwoAgents-GeneralGraph-MonotoneVals-Existence}, we can restrict attention to identical monotone valuations over a set of goods. Let $L$ and $R$ be the bipartition in the conflict graph and assume without loss of generality that $v\paren{ L } \geq v\paren{ R }$ and that all isolated vertices in the graph are in $L$. It follows that $L$ is maximal independent set in the conflict graph. Observe that running \cref{algo:sequence-from-good-S} with input $L$ ensures that $X_{ \minconf }, X_{ \maxconf } \subseteq \overline{L} = R$ which by monotonicity means that $\max \set{ v\paren{ X_{ \minconf } }, v\paren{ X_{ \maxconf } } } \leq v\paren{ R } \leq v\paren{ L }$. As \cref{algo:sequence-from-good-S} is efficient, we are done by \cref{lemma:spropforproof}.
\end{proof}

\subsubsection{Interval Graphs}
\label{sec:two-agents:specialized:interval}

We now consider the case where the conflict graph is an interval graph. Recall that an interval graph is an undirected graph where each vertex can be associated with an interval on the real line, such that two vertices are adjacent if and only if their corresponding intervals overlap (i.e., have a non-empty intersection). 

\begin{restatable}{theorem}{TwoAgents-IntervalGraph-MonotoneVals-Algorithm}
For two agents with monotone valuation functions and any interval graph as the conflict graph, a maximal and \EF{1} allocation exists and can be computed in polynomial time.
\label{thm:TwoAgents-IntervalGraph-MonotoneVals-Algorithm}
\end{restatable}

In order to prove \cref{thm:TwoAgents-IntervalGraph-MonotoneVals-Algorithm}, we first recall the well-known $\paren{ m, c }$-interval scheduling problem. In this problem, there are $m$ intervals (of the real line) and a parameter $c > 0$ and a collection of intervals is said to be $c$-feasible if any point appears in at most $c$ intervals. The goal is to find the largest $c$-feasible collection of intervals. When $c=1$, this is the classical interval scheduling problem, for which the earliest-finish-time greedy algorithm is optimal~\citep{KT06algorithm}. For general $c$, the greedy algorithm in \cref{algo:interval-scheduling} is also known to be optimal; see, e.g.,~\citep{CL95k}.
Here, we fix $m, c > 0$ and consider intervals $\mathcal{I}_i = \oc{ \ell_i, r_i }$ for all $i \in [m]$ that are sorted in non-decreasing order of $r_i$-s. We assume, without loss of generality, that all endpoints are distinct. For a set $S \subseteq [m]$ of intervals, we define the notation $r_S = \max_{ i \in S } r_i$. If $S = \emptyset$, we define $r_S = - \infty$.

\begin{algorithm}[H]
\caption{Greedy algorithm for the $\paren{ m, c }$-interval scheduling problem.}
\label{algo:interval-scheduling}
\begin{algorithmic}[1]

\renewcommand{\algorithmicrequire}{\textbf{Input:}}
\renewcommand{\algorithmicensure}{\textbf{Output:}}

\alglinenoPop{algcommon}

\smallskip

\Require Intervals $\mathcal{I}_i = \oc{ \ell_i, r_i }$ for all $i \in [m]$. We assume $r_i < r_{ i' }$ for all $i < i' \in [m]$.
\Ensure A $c$-feasible set $S = S_1 \cup \dots \cup S_c$, where the sets $S_1, \dots, S_c$ are disjoint and $1$-feasible.

\smallskip

\State Initialize sets $S_1 = \dots = S_c = \emptyset$. 

\smallskip

\For{$i \in [m]$}

\smallskip

\State If $\ell_i \geq \min_{ j \in [c] } r_{ S_j }$, add $i$ to $S_{ j^* }$, where $j^* = \argmax_{ j : \ell_i \geq r_{ S_j } } r_{ S_j }$ (ties broken arbitrarily).\label{line:add-to-set}

\smallskip

\EndFor

\smallskip

\alglinenoPush{algcommon}

\end{algorithmic}
\end{algorithm}

It is easy to see that the sets $S_1, \dots, S_c$ are always disjoint. Note from \cref{line:add-to-set} that whenever our algorithm adds an element $i$ to a set $S_{ j^* }$, we have that $\ell_i \geq r_{ S_{ j^* } }$. This ensures that the set remains $1$-feasible implying that the union of the $c$ sets is $c$-feasible. 

We now show that this procedure is optimal. For all $i \in [m]$ and $j \in [c]$, we let $S_j^i$ to be the value of $S_j$ at the end of iteration $i$ and let $S^i = S_1^i \cup \dots \cup S_c^i$. Also, define $S^0 = S_1^0 = \dots = S_c^0 = \emptyset$ to be the values at the beginning of the loop. We sometimes omit the superscript $i$ when $i = m$. 

\begin{lemma}
\label{lemma:min-collision}
For all $i \in \set{ 0 } \cup [m]$ and $j', j'' \in [c]$ such that $S_{ j' }^i \neq \emptyset$ and $r_{ S_{ j' }^i } \leq r_{ S_{ j'' }^i }$, there exists $i' \in S_{ j'' }^i$ such that $r_{ S_{ j' }^i } \in \mathcal{I}_{ i' }$.
\end{lemma}
\begin{proof}
Proof by induction on $i$. The base case $i = 0$ is trivial. We will show the lemma for $i > 0$ assuming it holds for $i - 1$. Assume first that $S_{ j' }^i \neq S_{ j' }^{ i - 1 }$, that is, the set $S_{ j' }$ changes in iteration $i$. Note that this means $S_{ j' }^i = S_{ j' }^{ i - 1 } \cup \set{i}$ and $S_{ j'' }^i = S_{ j'' }^{ i - 1 }$ implying that $r_{ S_{ j' }^i } = r_i \geq r_{ S_{ j'' }^i }$. This means that $r_{ S_{ j' }^i } = r_{ S_{ j'' }^i }$ implying from the distinctness of the endpoints that $j' = j''$ and the lemma follows. 

Assume now that $S_{ j' }^i = S_{ j' }^{ i - 1 }$. As the result is trivial otherwise, we can also assume that $r_{ S_{ j'' }^{ i - 1 } } < r_{ S_{ j' }^i } = r_{ S_{ j' }^{ i - 1 } }$. With this assumption, note that $r_{ S_{ j' }^i } \leq r_{ S_{ j'' }^i }$ is possible only if $S_{ j'' }^i = S_{ j'' }^{ i - 1 } \cup \set{i}$. From \cref{line:add-to-set} of \Cref{algo:interval-scheduling}, this means that $j'' = \argmax_{ j : \ell_i \geq r_{ S_j^{ i - 1 } } } r_{ S_j^{ i - 1 } }$. Combining with $r_{ S_{ j'' }^{ i - 1 } } < r_{ S_{ j' }^{ i - 1 } }$, this is possible only if $\ell_i < r_{ S_{ j' }^{ i - 1 } } = r_{ S_{ j' }^i }$. As $r_{ S_{ j' }^i } \leq r_i$, we are done by setting $i' = i$.
\end{proof}

\begin{lemma}
\label{lemma:greedy-earlier}
Let $S'$ be any $c$-feasible set and $i' \in S' \setminus S$ be arbitrary. Define $j' = \argmin_{ j \in [c] } r_{ S_j^{ i' - 1 } }$. Then, we have $S_{ j' }^{ i' - 1 } \neq \emptyset$. Additionally, letting $i'' = \max\paren{ S_{ j' }^{ i' - 1 } }$, we have:
\[
\card{ S' \cap \bracket{ i'', i' } } \leq \card{ S \cap \bracket{ i'', i' } } .
\]
\end{lemma}
\begin{proof}
Note that $S_{ j' }^{ i' - 1 } \neq \emptyset$ directly follows from \cref{line:add-to-set} of \Cref{algo:interval-scheduling} and the fact that $i' \notin S$. This means that $i''$ is well defined. Also, the fact that $r_i$-s are non-decreasing implies that $r_{ i'' } = r_{ S_{ j' }^{ i' - 1 } }$. Applying \cref{lemma:min-collision} with $i = i' - 1$ and $j'$, we get that for all $j''$, there exists $i''_{ j'' } \in S_{ j'' }^{ i' - 1 }$ such that $r_{ i'' } \in \mathcal{I}_{ i''_{ j'' } }$. As we assume the endpoints of all intervals are distinct and that $r_i$-s are non-decreasing, we get that $i'' \leq i''_{ j'' } < i'$.

Now, fix an arbitrary $i^* \in \bracket{ i'', i' }$ such that $r_{ i'' } \notin \mathcal{I}_{ i^* }$. We claim that $i^* \in S$. Indeed, from $r_{ i'' } \notin \mathcal{I}_{ i^* }$ and the fact that $r_i$-s are non-decreasing, we get that $\ell_{ i^* } \geq r_{ i'' } = r_{ S_{ j' }^{ i' - 1 } }$. As the sets $S_1, \dots, S_c$ can only get larger, we get that $\ell_{ i^* } \geq r_{ S_{ j' }^{ i^* - 1 } }$. From \cref{line:add-to-set}, we get that $i^* \in S$, as desired. With this claim, we have:
\begin{align*}
\card{ S' \cap \bracket{ i'', i' } } &= \card{ S' \cap \bracket{ i'', i' } \cap \set{ i''' \mid r_{ i'' } \in \mathcal{I}_{ i''' } } } + \card{ S' \cap \bracket{ i'', i' } \cap \set{ i''' \mid r_{ i'' } \notin \mathcal{I}_{ i''' } } } \\
&\leq \card{ S' \cap \set{ i''' \mid r_{ i'' } \in \mathcal{I}_{ i''' } } } + \card{ \bracket{ i'', i' } \cap \set{ i''' \mid r_{ i'' } \notin \mathcal{I}_{ i''' } } } \\
&\leq c + \card{ \bracket{ i'', i' } \cap \set{ i''' \mid r_{ i'' } \notin \mathcal{I}_{ i''' } } } \tag{As $S'$ is $c$-feasible} \\
&= c + \card{ S \cap \bracket{ i'', i' } \cap \set{ i''' \mid r_{ i'' } \notin \mathcal{I}_{ i''' } } } \tag{By our claim} \\
&= c + \card{ S \cap \bracket{ i'', i' } } - \card{ S \cap \bracket{ i'', i' } \cap \set{ i''' \mid r_{ i'' } \in \mathcal{I}_{ i''' } } } \\
&\leq \card{ S \cap \bracket{ i'', i' } } . \tag{As $i''_1, \dots, i''_c \in S \cap \bracket{ i'', i' } \cap \set{ i''' \mid r_{ i'' } \in \mathcal{I}_{ i''' } }$ are all distinct}
\end{align*}
\end{proof}

\begin{lemma}
\label{lemma:greedy-optimal}
For any $c$-feasible solution $S'$ and any $i' \in [m]$, we have that $\card{ S' \cap \bracket{ i' } } \leq \card{ S \cap \bracket*{ i' } }$. Taking $i' = m$, it follows that $S$ is an optimal solution for the $\paren{ m, c }$-interval scheduling problem with inputs $\paren{ \mathcal{I}_i }_{ i \in [m] }$. 
\end{lemma}
\begin{proof}
Proof by contradiction. Suppose there exists another $c$-feasible solution $S'$ and $i' \in [m]$ such that $\card{ S' \cap \bracket{ i' } } > \card{ S \cap \bracket*{ i' } }$. Pick the smallest such $i'$. Note that $i' > 1$ as \cref{algo:interval-scheduling} ensures that $1 \in S$. By our choice of $i'$, we have that $i' \in S' \setminus S$. By \cref{lemma:greedy-earlier}, we have $i'' < i'$ such that $\card{ S' \cap \bracket{ i'', i' } } \leq \card{ S \cap \bracket{ i'', i' } }$. This yields the following contradiction:
\begin{align*}
\card{ S' \cap \bracket{ i' } } &= \card{ S' \cap \bracket{ i'' - 1 } } + \card{ S' \cap \bracket{ i'', i' } } \\
&\leq \card{ S \cap \bracket{ i'' - 1 } } + \card{ S \cap \bracket{ i'', i' } } \tag{Choice of $i'$} \\
&\leq \card{ S \cap \bracket{ i' } } .\qedhere
\end{align*}
\end{proof}

In addition to being optimal, the solution output by \cref{algo:interval-scheduling} also satisfies the following ``exchange lemma'' when $c = 1$. 
\begin{lemma}
\label{lemma:exchange}
Let $c = 1$ and $S = \set{ i_1, \dots, i_k }$ be the (sorted) set output by \cref{algo:interval-scheduling}. Let $S' = \set{ i'_1, \dots, i'_{ k' } }$ be any other (sorted) $c$-feasible set\footnote{The optimality of $S$ implies that $k' \leq k$.}. Then, for all $0 \leq k'' \leq k'$, the set $S''_{ k'' } = \set{ i_1, \dots i_{ k'' }, i'_{ k'' + 1 }, \dots, i'_{ k' } }$ is $c$-feasible.
\end{lemma}
\begin{proof}
Proof by induction on $k''$. The base case $k'' = 0$ is trivial. We show the result for $k'' > 0$ assuming it holds for $k'' - 1$. Assume not. Then by the induction hypothesis, we have that the set $S''_{ k'' - 1 }$ is $c$-feasible but the set $S''_{ k'' }$ is not. As $S$ is $c$-feasible and $c = 1$, this means that there exists $k'' > k''$ such that $\mathcal{I}_{ i_{ k'' } } \cap \mathcal{I}_{ i'_{ k''' } } \neq \emptyset$. With this, the $c$-feasibility of $S'$ implies that $i'_{ k'' } < i_{ k'' }$. This means that $\card{ S' \cap \bracket{ i'_{ k'' } } } = k'' > k'' - 1 \geq \card{ S \cap \bracket*{ i'_{ k'' } } }$, contradicting \cref{lemma:greedy-optimal}.
\end{proof}

We are now ready to prove \cref{thm:TwoAgents-IntervalGraph-MonotoneVals-Algorithm}.
\begin{proof}[Proof of \cref{thm:TwoAgents-IntervalGraph-MonotoneVals-Algorithm}]
As in the proof of \cref{thm:TwoAgents-GeneralGraph-MonotoneVals-Existence}, we can restrict attention to identical monotone valuations over a set of goods. Given that we are working with an interval graph, each good $o \in [m]$ can be associated with an interval~$\oc{ \ell_o, r_o }$ of the real line such that two goods conflict if and only if their intervals intersect. We assume without loss of generality that the goods are sorted in non-decreasing order of the $r_o$-s. Let $Z = \set{ z_1, \dots, z_k } \subseteq [m]$ be the output of \Cref{algo:interval-scheduling} when $c = 2$, and let $Z_1$ and $Z_2$ denote the $1$-feasible subsets constructed by the algorithm such that $Z = Z_1 \cup Z_2$. Observe that both $Z_1$ and $Z_2$ are independent subsets of the original conflict graph $G$. Swap $Z_1$ and $Z_2$ if needed to ensure that $v\paren{ Z_1 } \geq v\paren{ Z_2 }$. 



We greedily move intervals from $Z_2$ to $Z_1$ if doing so does not violate the fact that $Z_1$ constitutes an independent set. Note that monotonicity implies that $Z_1$ and $Z_2$ are independent sets and that $v\paren{ Z_1 } \geq v\paren{ Z_2 }$ even after these moves. 

We claim that the modified $Z_1$ is a maximal independent set in the conflict graph $G$. Indeed, any element of $Z_2$ cannot be added to it while preserving independence by construction. Additionally, if an element outside $Z$ could be added, the fact that $Z_1$ and $Z_2$ are independent sets would violate the optimality of $Z$. 

Having shown that $Z_1$ is a maximal independent set in the conflict graph, we can now invoke \cref{algo:sequence-from-good-S} with $S = Z_1$ and the goods ordered in increasing order of $r_o$, their finish times. Let $X_{ \minconf }$ and $X_{ \maxconf }$ be the sets constructed in \cref{line:x-max-min-conf} of \Cref{algo:sequence-from-good-S}. Observe that $X_{ \minconf }$ and $X_{ \maxconf }$ are 
outputs of the greedy algorithm applied to the set $\overline{ Z_1 }$ where the vertices are indexed according to the finish times of their corresponding intervals. Thus, by applying~\cref{lemma:greedy-optimal} to the vertices in $\overline{ Z_1 }$ with $c = 1$, it follows that $X_{ \minconf }$ and $X_{ \maxconf }$ are the largest $1$-feasible subsets of $\overline{ Z_1 }$. As $Z_2$ is also $1$-feasible subset of $\overline{ Z_1 }$, we have that $\card{ Z_2 } \leq \min \set{ \card{ X_{ \minconf } }, \card{ X_{ \maxconf } } }$.

Moreover, the fact that $Z_1$ and $X_{ \minconf }$ (respectively, $X_{ \maxconf }$) are both maximal independent sets implies that $Z_1 \cup X_{ \minconf }$ (respectively, $Z_1 \cup X_{ \maxconf }$) is $2$-feasible. 
By the optimality of $Z$, we get that $\card{ Z_2 } \geq \max \set{ \card{ X_{ \minconf } }, \card{ X_{ \maxconf } } }$. Combining, we get that $\card{ Z_2 } = \card{ X_{ \minconf } } = \card{ X_{ \maxconf } }$. Using these equalities and \cref{lemma:exchange}, we get that there exists a sequence of allocations from the allocation $\paren{ Z_1, Z_2 }$ to the allocation $\paren{ Z_1, X_{ \minconf } }$ and another sequence of allocations from the allocation $\paren{ X_{ \maxconf }, Z_1 }$ to the allocation $\paren{ Z_2, Z_1 }$ such that in both sequences, every allocation is maximal and every pair of consecutive allocations is order-adjacent.

Inserting the sequence of allocations constructed in \cref{line:ao} between these sequences, we get from \cref{lemma:ao-maximal,lemma:ao-order-adjacent} that there is a sequence of allocations from $\paren{ Z_1, Z_2 }$ to $\paren{ Z_2, Z_1 }$ where every allocation is maximal and every pair of consecutive allocations is order-adjacent. As $v\paren{ Z_1 } \geq v\paren{ Z_2 }$, we get that there must be a pair of consecutive allocations in this sequence satisfying the conditions of \cref{lemma:EF1-from-order-adjacent} which then implies the theorem.
\end{proof}

%
%

\subsubsection{Additive Valuations}
\label{sec:two-agents:specialized:additive}

Now, we show that whenever the valuations are additive, then, for any conflict graph $G$, a maximal and \EF{1} allocation exists and can be computed in polynomial time. 

\begin{restatable}{theorem}{TwoAgents-GeneralGraph-AdditiveVals-Algorithm}
For two agents with additive and monotone valuation functions and any conflict graph, a maximal and \EF{1} allocation exists and can be computed in polynomial time.
\label{thm:TwoAgents-GeneralGraph-AdditiveVals-Algorithm}
\end{restatable}
\begin{proof}
As in the proof of \cref{thm:TwoAgents-GeneralGraph-MonotoneVals-Existence}, we can restrict attention to identical additive valuations over a set of goods. Let $m$ denote the number of items and $v$ denote the valuation functions of the agents. We first claim that running \cref{algo:sequence-from-good-S} for a given maximal independent set $S$ either produces a maximal and $\EF{1}$ allocation or produces a maximal independent set $S'$ satisfying $v\paren{ S' } \geq \frac{ m }{ m - 1 } \cdot v\paren{ S }$. Indeed, by \cref{lemma:ao-maximal,lemma:spropforproof}, we have that if \cref{algo:sequence-from-good-S} does not produce a maximal and $\EF{1}$ allocation, we have that $v\paren{ S } < \max\{ v\paren{ X_{ \minconf } }, v\paren{ X_{ \maxconf } } \}$. As the other case is similar, assume that $v\paren{ S } < v\paren{ X_{ \maxconf } }$. Since, we know that $\mathcal{A}_m = \paren{ X_{ \maxconf }, S }$ is not \EF{1}, we can deduce that for all $o' \in X_{ \maxconf }$, it holds that $v\paren{ S } < v\paren{ X_{ \maxconf } \setminus \set{ o' } }$. Taking $o'$ to be the item with the largest valuation, we have that:
\[
v\paren{ S } < v\paren{ X_{ \maxconf } \setminus \set{ o' } } = v\paren{ X_{ \maxconf } } - v\paren{ \set{ o' } } \leq v\paren{ X_{ \maxconf } } - \frac{ v\paren{ X_{ \maxconf } } }{ m } .
\]
The claim follows by setting $S' = X_{ \maxconf }$. With this claim in hand, we find the required allocation by repeatedly invoking \cref{algo:sequence-from-good-S} for different values of $S$, starting from any maximal independent set $S_0$ in the conflict graph that contains the item $o^*$ with the highest valuation. Such an $S_0$ can be efficiently obtained by a greedy procedure. Note that additivity implies that $v\paren{ S } \leq m \cdot v\paren{ S_0 }$ for any set $S$. Now, use our claim above to show that starting with the set $S_0$ would either produces a maximal and $\EF{1}$ allocation or produces a maximal independent set $S'$ whose valuation is a factor of $\frac{ m }{ m - 1 }$ larger. As $v\paren{ S } \leq m \cdot v\paren{ S_0 }$ for any set $S$, the latter cannot happen more that $O\paren{ m \log m }$ times. Thus, in polynomial time, we will find a maximal and \EF{1} allocation.
\end{proof}

So far, we have focused on the case of two agents with monotone valuations. 
The monotonicity assumption is essential: once valuations are allowed to be additive but not necessarily monotone, a maximal \EF{1} allocation may fail to exist, even for two agents with identical additive valuations.

\begin{restatable}
{theorem}{NonExistence-Additive}
For two agents on a path, there exists an instance with identical additive non-monotone valuations where no
maximal \EF{1} allocation exists.
\label{thm:NonExistence-Additive}
\end{restatable}
\begin{proof}
    Consider a path graph with two items: one good valued at $+1$ and one chore valued at $-1$ by both agents. Any maximal allocation must assign both items, with one assigned to each agent. However, this assignment violates \EF{1}.
\end{proof}

In the next section, we will extend this counterexample to any fixed $n \geq 2$ (\Cref{prop:coutnerexample:mixed}).

\section{Results for An Arbitrary Number of Agents}
\label{sec:Results_n_Agents}

We will now turn to the case of an arbitrary number of agents. We will start by showing that an \EF{1} and maximal allocation can fail to exist for three agents with identical monotone valuations (\Cref{thm:NonExistence-Monotone}), or for $n \geq 4$ agents with identical non-negative additive valuations (\Cref{thm:counterexample_n4}). Additionally, we will provide a general framework that transforms any counterexample for \EF{1} and maximality into an \NPH{}ness proof (\Cref{thm:NPhardness}) of checking the existence of such allocations. As a consequence, we will show that determining the existence of an \EF{1} and maximal allocation is \NPH{} for (a) three agents with identical monotone valuations, (b) any fixed $n \geq 2$ agents, with identical additive valuations, and (c) any fixed $n \geq 4$ agents, with identical non-negative additive valuations.

Given the negative existential and computational results, we restrict our attention to identical additive valuations and a path graph. In this setting, we show that an allocation that is envy-free up to one good and one chore (\EFoneone{}) and maximal always exists (\Cref{thm:n_Agents_Identical}). Our proof uses a novel application of the ``cycle plus triangles'' theorem, specifically, its generalization known as ``cycle plus $n$-cliques'' theorem, from graph theory.

We also consider the case of \emph{uniform valuations} in which each item has the same marginal value for every set. We show that when the conflict graph is a tree, a maximal \EF{1} allocation is guaranteed to exist and can be found in polynomial time~(\Cref{thm:uni-tree2}).

\subsection{General Valuations on General Graphs}

\begin{restatable}
{theorem}{NonExistence-Monotone}
For $n = 3$ agents, there exists an instance with identical monotone valuations where no
maximal \EF{1} allocation exists.
\label{thm:NonExistence-Monotone}
\end{restatable}
\begin{proof}
    Consider the following instance with $7$ goods. The graph consists of a complete bipartite graph $K_{3,3}$ with bipartition $X=\{o_1, o_2, o_3\}$ and $Y=\{o_4, o_5, o_6\}$ together with two edges $\{o_1, o_7\}$ and $\{o_4, o_7\}$. Each of the three agents has an identical monotone valuation $v$ such that 
\begin{itemize}
			\item $v(\emptyset) = 0$; 
			\item $v(S) = 1$ if $|S|=1$ and $S\in \{ \{o_1\},\{o_4\} \}$; 
			\item $v(S) = 2$ if $|S|=1$ and $S \not \in \{ \{o_1\},\{o_4\} \}$; 
			\item $v(S) = 3$ if $S$ is $\{o_2, o_7\}$, $\{o_3, o_7\}$, $\{o_5, o_7\}$, or $\{o_6, o_7\}$; 
			\item $v(S) = 4$ for any other $S \subseteq M$. 
\end{itemize}

	Note that the symmetry on the graph and valuation holds: swapping $o_2$ and $o_3$, swapping $o_5$ and $o_6$, and swapping $\{o_1, o_2, o_3\}$ and $\{o_4, o_5, o_6\}$ all lead to the same instance, and swapping the bundles of agents does not change whether the allocation is EF1 or not because valuations are identical. 

Specifically, consider any maximal allocation $\mathcal{A}$.
By maximality, good $o_7$ must be allocated to some agent since it has degree $2$. Since agents have identical valuations, without loss of generality, suppose it is allocated to agent $1$. Thus, agent $1$ can receive neither $o_1$ nor $o_4$; let us assume that if $o_1$ is allocated, then it is allocated to agent $2$ and if $o_4$ is allocated, then it is allocated to agent $3$. Further, agent $1$ can receive a good from at most one part of the bipartition $X$ and $Y$ due to conflicting constraints. Without loss of generality, suppose it is $Y$. Now consider the following cases. 

\begin{enumerate}
\item If none of the goods in $Y$ is allocated to agent $1$, then $\mathcal{A}$ must allocate $X$ to agent $2$ and $Y$ to agent $3$. 
\item Suppose that exactly one of the goods in $Y$, say good $o_5$, is allocated to agent $1$. In this case, observe that $o_4$ must be allocated. This is because if $o_4$ is unallocated, this means that both agents $2$ and $3$ receive at least one good from $X$ and $o_6$ must be allocated to agent $1$, a contradiction. 
Thus, the only possible maximal allocations are $(\{o_5,o_7\},\{o_1,o_2,o_3\},\{o_4,o_6\})$ and $(\{o_5,o_7\},\{o_6\},\{o_4\})$.
A similar argument holds when $o_6$ is allocated to agent $1$ due to the symmetry of the graph. 

\item Suppose that two goods, $o_5$ and $o_6$, in $Y$ are allocated to agent $1$. If $o_4$ is allocated, then the only possible maximal allocation is $(\{o_5,o_6,o_7\},\{o_1,o_2,o_3\},\{o_4\})$. 
If $o_4$ is unallocated, this means that both agents $2$ and $3$ receive at least one good from $X$ and all goods in $X$ must be allocated to either of such agents. Thus, the only possible maximal allocations are $(\{o_5,o_6,o_7\},\{o_1\},\{o_2,o_3\})$, $(\{o_5,o_6,o_7\},\{o_1,o_2\},\{o_3\})$, together with $(\{o_5,o_6,o_7\},\{o_1,o_3\},\{o_2\})$. 
\end{enumerate}

    For this reason, there are only $6$ maximal allocations to consider (refer to Figure \ref{fig:counterexample_n3}). All of them are not EF1, because:
	\begin{itemize}
		\item Allocations \#1, \#4, \#5, \#6: one agent takes one good but there is another agent who takes three goods
		\item Allocation \#2: $v(\{o_5, o_7\}) = 3$ but $v(\{o_1, o_2\}) = v(\{o_1, o_3\}) = v(\{o_2, o_3\}) = 4$.
		\item Allocation \#3: $v(o_4) = 1$ but $v(o_5) = v(o_7) = 2$.~\qedhere
	\end{itemize}
\end{proof}


\newcommand{\mygraphC}[4]{%
\begin{scope}[xshift=#2,yshift=#3]

\def\colors{#4}

\foreach \col [count=\i] in \colors {

    \ifnum\i=1 \node[vertex,fill=\col] (w#1\i) at (-1,0.75) {$o_1$};\fi
    \ifnum\i=2 \node[vertex,fill=\col] (w#1\i) at (0,0.75) {$o_2$};\fi
    \ifnum\i=3 \node[vertex,fill=\col] (w#1\i) at (1,0.75) {$o_3$};\fi
    \ifnum\i=4 \node[vertex,fill=\col] (w#1\i) at (-1,-0.75) {$o_4$};\fi
    \ifnum\i=5 \node[vertex,fill=\col] (w#1\i) at (0,-0.75) {$o_5$};\fi
    \ifnum\i=6 \node[vertex,fill=\col] (w#1\i) at (1,-0.75) {$o_6$};\fi
    \ifnum\i=7 \node[vertex,fill=\col] (w#1\i) at (-1.75,0) {$o_7$};\fi

}

\draw (w#11)--(w#14)--(w#12)--(w#15)--(w#13)--(w#16)--(w#12);
\draw (w#14)--(w#17)--(w#11)--(w#15)--(w#13);
\draw (w#11)--(w#16);
\draw (w#13)--(w#14);

\end{scope}
}

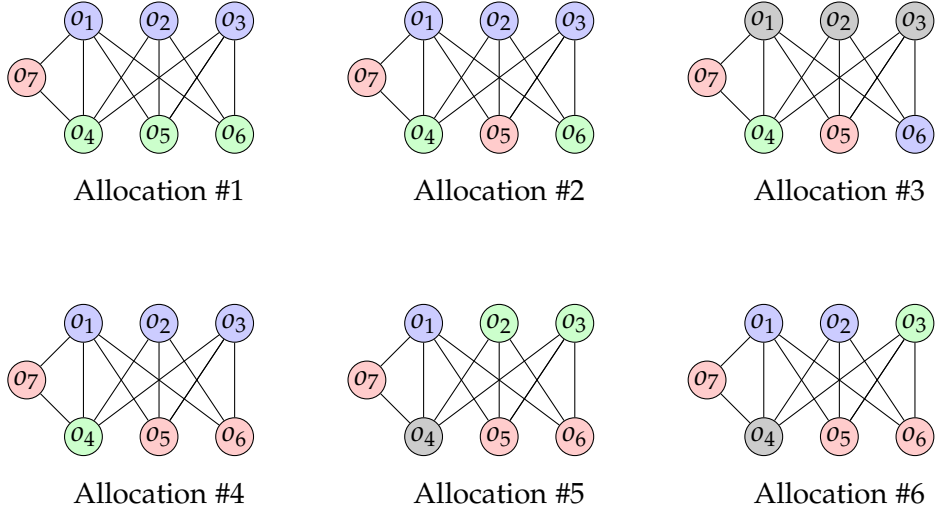
\begin{figure}[!htbp]
\centering

\begin{tikzpicture}[
    vertex/.style={
        circle,
        draw,
        minimum size=5mm,
        inner sep=0pt
    },
    middle/.style={midway, inner sep = 1pt, fill=white}
]


\mygraphC{1}{-4.5cm}{0cm}{blue!20, blue!20, blue!20, green!20, green!20, green!20, red!20}
\node[below=0.2cm of w15] {Allocation \#1};

\mygraphC{2}{0cm}{0cm}{blue!20, blue!20, blue!20, green!20, red!20, green!20, red!20}
\node[below=0.2cm of w25] {Allocation \#2};

\mygraphC{3}{4.5cm}{0cm}{black!20, black!20, black!20, green!20, red!20, blue!20, red!20}
\node[below=0.2cm of w35] {Allocation \#3};

\mygraphC{4}{-4.5cm}{-4cm}{blue!20, blue!20, blue!20, green!20, red!20, red!20, red!20}
\node[below=0.2cm of w45] {Allocation \#4};

\mygraphC{5}{0cm}{-4cm}{blue!20, green!20, green!20, black!20, red!20, red!20, red!20}
\node[below=0.2cm of w55] {Allocation \#5};

\mygraphC{6}{4.5cm}{-4cm}{blue!20, blue!20, green!20, black!20, red!20, red!20, red!20}
\node[below=0.2cm of w65] {Allocation \#6};

\end{tikzpicture}

\caption{$6$ maximal allocations to consider in the instance of Theorem~\ref{thm:NonExistence-Monotone}. The vertices in red, blue, green, and gray are goods taken by agent $1, 2, 3$, and no one, respectively.}
	\label{fig:counterexample_n3}

\end{figure}

The instance constructed in the proof of \Cref{thm:NonExistence-Monotone} uses an identical monotone valuation. It remains an open question whether a counterexample exists for three agents with (even identical) additive and monotone valuations on general graphs. 

For every number $n \geq 4$ of agents, \citet{HH22fair} presented an instance with identical additive valuations and a complete bipartite graph $K_{n-1,n-1}$ for which a complete EF1 allocation does not exist.
In such cases, a complete allocation does exist, as the maximum degree of the graph is less than the number of agents. Thus, this counterexample implies that achieving both EF1 and maximality is impossible even when four agents have identical additive valuations. Here, we provide a smaller example using $K_{3,n-1}$, which turns out to be smallest since for $m\leq n+1$, there always exists a maximal EF1 allocation. 


\begin{restatable}{proposition}{ExampleFourAgents}
 For every number $n \geq 4$ of agents, there is an instance with identical non-negative additive valuations, $m=n+2$, and $G=K_{3,n-1}$ where no maximal EF1 allocation exists.
    \label{thm:counterexample_n4}
\end{restatable}
\begin{proof}
    Consider the following instance with $m = n+2$ goods and a complete bipartite graph $K_{3, n-1}$ with the left part containing the three goods $\{o_1,o_2,o_3\}$ and the right part containing the $n-1$ goods $\{o_4,\dots,o_{n+2}\}$. Each agent has an identical additive valuation where
        \begin{equation*}
            v(o) = \begin{cases}
                2 & (o \in \{o_1, o_2, o_3\}) \\
                3 & (\text{otherwise}).
            \end{cases}
        \end{equation*}
    Consider an arbitrary maximal allocation $\mathcal{A}$. 
    Since the maximum degree of the graph is $n-1$, all goods must be allocated in this allocation. By the pigeonhole principle, there exists an agent $i$ that does not receive any of the goods on the right.
    \begin{itemize}
        \item Case 1: If there is exactly one such agent, then the goods $o_4,\dots,o_{n+2}$ are allocated to the other $n-1$ agents, which means that each of those agents receives one good of value $3$. The remaining goods $o_1,o_2,o_3$ must therefore be assigned to agent $i$. Thus, every other agent envies agent $i$ even after the removal of one of the goods $o_1,o_2,o_3$ from agent $i$'s bundle.
        \item Case 2: If there are two or more such agents, then there exists an agent $j$ who receives two or more goods among $o_4,\dots,o_{n+2}$. Further, there exists an agent $k$ who receives none of the goods on the right and at most one good from the left side. Then agent $k$ envies agent $j$ even after the removal of one good from agent $j$'s bundle.
    \end{itemize}
    Therefore, in either case, the allocation $\mathcal{A}$ is not EF1.
\end{proof}

\begin{proposition}
For every number $n \geq 2$ of agents, there is an instance with identical additive valuations, consisting of $n-1$ goods and one chore, where no maximal EF1 allocation exists.
\label{prop:coutnerexample:mixed}
\end{proposition}
\begin{proof}
Consider a complete graph on $n$ items. Let $o_1,\dots,o_{n-1}$ be goods of value $1$ and let $o_n$ be a chore of value $-1$, where all agents have the same additive valuation.

Since the graph is complete, every independent set has size at most one. Hence, in any maximal allocation, every agent must receive exactly one item, and therefore all $n$ items are allocated. In particular, one agent receives the chore $o_n$, while each of the remaining $n-1$ agents receives one good.

Let $i$ be the agent receiving $o_n$ and let $j$ be an agent receiving some good $o_k$ with $k \in [n-1]$. Then $v(A_i)=-1$ and $v(A_j)=1$. Since $A_i \cup A_j=\{o_n,o_k\}$, the only sets $S \subseteq A_i \cup A_j$ of size at most one are $\emptyset$, $\{o_n\}$, and $\{o_k\}$. For these choices, we obtain
\[
(v(A_i \setminus S), v(A_j \setminus S)) \in \{(-1,1), (0,1), (-1,0)\},
\]
so the inequality $v(A_i \setminus S) \ge v(A_j \setminus S)$ never holds. Thus, the allocation is not EF1. Since every maximal allocation has this form, no maximal EF1 allocation exists.
\end{proof}

\begin{restatable}
{theorem}{NPhardness}
Determining the existence of an \EF{1} and maximal allocation is \NPH{} even for:
\begin{itemize}
    \item three agents with identical monotone valuations, 
    \item any fixed $n \geq 2$ agents, with identical additive valuations, or
    \item any fixed $n \geq 4$ agents, with identical non-negative additive valuations.
\end{itemize}
\label{thm:NPhardness}
\end{restatable}

Now, we will show that it is NP-hard to decide whether a maximal EF1 allocation exists, as stated below.

In fact, our proof can transform any negative example into an NP-hardness result, as stated more precisely below.

\begin{lemma} \label{lem:main-red}
Suppose that there exists an instance $\tI = ([n], \tM, \tv, \tG = (\tM, \tE))$ (with identical valuation) where the number $n$ of agents and the number $|\tM|$ of items are both constants, such that no maximal EF1 allocation exists. Then, it is NP-hard to decide whether a maximal EF1 allocation exists for $n$ agents with identical valuations. 
Furthermore, if $\tv$ is additive, then this NP-hardness applies even when the valuations are additive. 
\end{lemma}

\Cref{thm:NPhardness} is an immediate corollary of  \Cref{lem:main-red} where $\tI$ is the instance from \Cref{thm:NonExistence-Monotone}, \Cref{thm:counterexample_n4}, or \Cref{prop:coutnerexample:mixed}. Note that we state \Cref{lem:main-red} in this generic form so that, if subsequent work finds such an instance $\tI$ for additive valuation for $n = 3$, then the NP-hardness would follow as a corollary.

\paragraph{Reduction.}
We will reduce from the Independent Set (IS) problem, which is one of Karp's classic NP-hard problems~\citep{Karp72}. In the IS problem, we are given a graph $H = (V_H, E_H)$ and a positive integer $t$, and the goal is to decide whether $H$ contains an IS of size $t$.

At a high-level, our reduction starts from $\tI$ and adds to it $n$ copies of the graph $H$, where each good has the same value $\lambda$. Roughly speaking, we wish the $i$-th copy of $H$ (denoted by $X_i$ in the proof below) to give ``extra goods'' to the $i$-th agent, in case that agent envies some other agent by more than one good. The crux of the reduction is that such an agent can ``catch up'' (and thus satisfy EF1) iff there is a sufficiently large independent set in $H$.  This is not yet a complete reduction since, $H$ may not have a \emph{maximal} independent set of a certain prescribed size. To alleviate this, we introduce ``dummy goods'' with zero value (denoted by $Y_i$ below) to ensure that we can pick any desired number of goods from each copy of $H$. Finally, some additional edges are also added to ensure that each agent selects goods from a single copy of $H$.

For the proof below, we will use the following notation: for any valuation $v$ and set $S$ of goods, let $v^{-1}(S) := \min_{j \in S} v(S \setminus \{j\})$ denote the value of $S$ after its most valuable good is removed. We use the convention $v^{-1}(\emptyset) = 0$.

\begin{proof}[Proof of \Cref{lem:main-red}]
Recall $\tI$ from the lemma statement.
Let\footnote{$\gamma$ can be computed in $O(1)$ time by bruteforce.}
$\gamma := \min_{\tcA} \max_{i, i' \in [n]} \left(\tv^{-1}(\tA_i) - \tv(\tA_{i'})\right)$ where the outer minimum is over all maximal allocation $\tcA$ of $\tI$. By the assumption on $\tI$, we have $\gamma > 0$. Let $\lambda := \gamma / t$.

Let $(H = (V_H, E_H), t)$ denote the IS instance. Our reduction constructs the instance $I = ([n], M, v, G)$ as follows:
\begin{itemize}
\item {\bf Goods}: $M = \tM \cup X_1 \cdots \cup X_n \cup Y_1 \cup \cdots \cup Y_n$ where $X_i = \{x_{i, w} \mid w \in V_H\}$ and $Y_i = \{y_{i, w} \mid w \in V_H\}$ are sets (each of size $|V_H|$) of additional goods.
\item {\bf Graph}: $G$ contains the following edges:
	\begin{enumerate}[(i)]
	\item All edges in $\tG$,
	\item $(x_{i, u}, x_{i, w})$ for all $i \in [n]$ and $(u, w) \in V_H$,
        \item $(x_{i, w}, y_{i, w})$ for all $i \in [n]$ and $w \in V_H$,
        \item all pairs of vertices in $(X_i \cup Y_i) \times (X_{i'} \cup Y_{i'})$ for all distinct $i, i' \in [n]$.
	\end{enumerate}
	\item {\bf Valuation}: For all $S \subseteq M$, let $v(S) = \tv(S \cap \tM) + \lambda |S \cap X|$ where $X := X_1 \cup \cdots \cup X_n$. That is, the valuations on original goods remain the same, the valuations of each good in $X$ is $\lambda$, and the goods in $Y_1 \cup \cdots \cup Y_n$ have valuations zero.
\end{itemize}

See Figure \ref{fig:nphard_construction} for an illustration of the instance $I$.



It is clear that the reduction runs in polynomial time, and that, if $\tv$ is additive, then $v$ is also additive.

\paragraph{(YES)} Suppose that $H$ contains an IS of size $t$.
Let $\tcA^*$ be a maximal allocation of $\tI$ such that $\max_{i, i' \in [n]} \left(\tv^{-1}(\tA^*_i) -  \tv(\tA^*_{i'})\right) = \gamma$. We may assume w.l.o.g. that $v^{-1}(\tA^*_1) \geq v^{-1}(\tA^*_2), \dots, \geq v^{-1}(\tA^*_n)$. For each $i \in [n]$, we construct $A^*_i$ as follows:
\begin{enumerate}
\item  Let $c_i := \lceil \max\{0, \tv^{-1}(\tA^*_1) - \tv(\tA_i^*)\} / \lambda \rceil \leq t$.
\item Let $S_i$ be any (non-necessarily maximal) IS of size $c_i$ in $H$, which exists since $H$ contains an IS of size $t$.
\item Let $A^*_i = \tA^*_i \cup \{x_{i, v}\}_{v \in S_i} \cup \{y_{i, v}\}_{v \in (V_H \setminus S_i)}$
\end{enumerate}
Observe that each good in $X_i \cup Y_i$ can only belong to $A^*_i$, and it is obvious that there is no edge between goods in $A^*_i$. Thus, $\cA^* = (A^*_1, \dots, A^*_n)$ is a valid allocation. To see that this is maximal, note that the goods from $Y_i$ (resp., $X_i$) together with type-(iii) edges ensure that no other goods in $X_i$ (resp., $Y_i$) can be added to $A^*_i$. Since at least one good from $X_i \cup Y_i$ is picked, type-(iv) edges ensure that no goods in $X_{i'} \cup Y_{i'}$ for $i' \ne i$ can be added to $A^*_i$.

Finally, we argue that $\cA^*$ is EF1. 
To bound $v^{-1}(A^*_i)$, note that $v(A^*_{i}) = \tv(\tA^*_{i}) + c_i \lambda$. Consider two cases based on $c_i$.
\begin{itemize}
\item If $c_i = 0$, we have $v^{-1}(A^*_i) = \tv^{-1}(\tA^{*}_i) \leq \tv^{-1}(\tA^*_1)$.
\item If $c_i > 0$, by definition of $c_i$, we have $v(A^*_i) < \tv^{-1}(\tA^*_1) + \lambda$. Thus, $v^{-1}(A^*_i) \leq v(A^*_i) - \lambda < \tv^{-1}(\tA^*_1)$.
 \end{itemize}
Thus, in both cases, we have $v^{-1}(A^*_i) \leq \tv^{-1}(\tA^*_1)$. 

On the other hand, for any $i' \in [n]$, the definition of $c_{i'}$ immediately implies $v(A^*_{i'}) \geq \tv^{-1}(\tA^*_1)$.

By the two previous paragraphs, $\cA^*$ is EF1.

\paragraph{(NO)} Suppose that $H$ does not contain an IS of size $t$. Consider any maximal allocation $\cA = (A_1, \dots, A_n)$ of $I$. Notice that the allocation $\tcA = (A_1 \cap \tM, \dots, A_n \cap \tM)$ is maximal w.r.t. $\tI$. Thus, there exist $i, i' \in [n]$ such that $\tv^{-1}(\tA_i) - \tv(\tA_{i'}) \geq \gamma$. Due to type-(iv) edges, at most one of $A_{i'} \cap X_1, \dots, A_{i'} \cap X_n$ can be non-empty. Furthermore, type-(ii) edges imply that the non-empty set must correspond to an independent set in $H$. From our assumption, this implies $|A_{i'} \cap X| < t$. As a result,
\begin{align*}
v^{-1}(A_i) \geq \tv^{-1}(A_i)
&\geq \gamma + \tv(\tA_{i'}) \\
&= \gamma + v(A_{i'}) - \lambda |A_{i'} \cap X| \\
&> \gamma + v(A_{i'}) - \lambda \cdot t  \qquad \overset{(\star)}{\geq} v(A_{i'}),
\end{align*}
where $(\star)$ is due to our choice of $\lambda$. Thus, $\tA$ is not EF1.
\end{proof}

\begin{figure}[!htbp]
	\centering
	\includegraphics[width=0.7\textwidth]{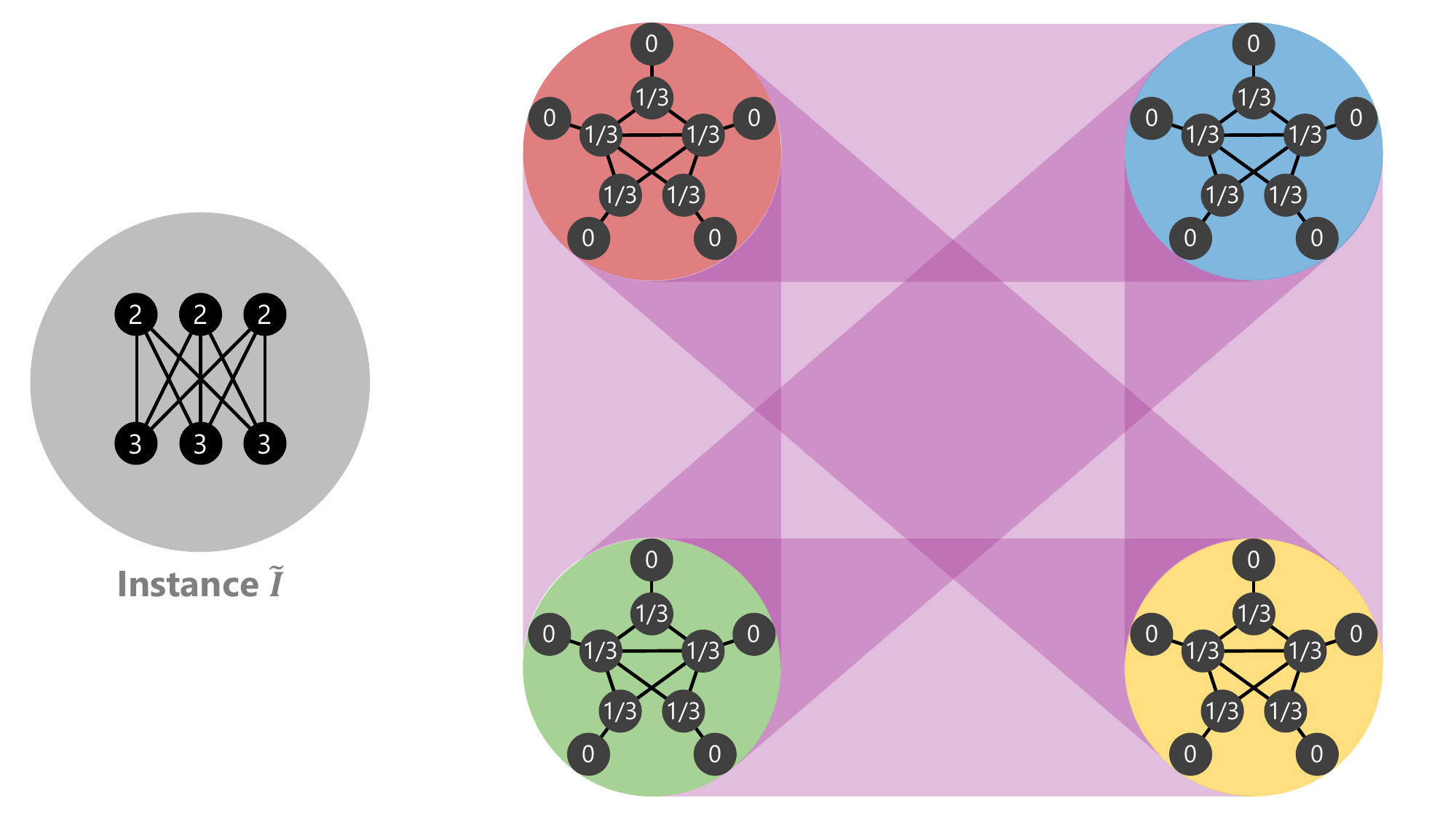}
	\caption{Instance $I$ created with $n = 4$ instance given by Proposition \ref{thm:counterexample_n4} for $\tI$, and $5$-vertex $7$-edge graph for $H$, setting $t = 3$ (note that $\gamma = 1, \lambda = \frac{1}{3}$). The bands in purple represents type-(iv) edges.}.
	\label{fig:nphard_construction}
\end{figure}


\subsection{Identical Additive Valuations on a Path}

We will now establish a positive result for \emph{identical} additive valuations when the conflict graph is a \emph{path}~(\Cref{thm:n_Agents_Identical}). 

\begin{restatable}
{theorem}{nAgentsIdentical}
Given any instance with $n \geq 3$ agents with identical additive valuations and an arbitrary path graph, an EF[1,1] and complete (therefore, maximal) allocation always exists. Moreover, if the valuations are monotone, an \EF{1} and complete allocation always exists.
\label{thm:n_Agents_Identical}
\end{restatable}

We will start by discussing the key ideas in the proof of \Cref{thm:n_Agents_Identical} for the case of $n=3$ agents. Later, we will extend this discussion to provide a detailed proof for general $n$. Our proof will use the ``cycle plus triangles'' theorem~(\Cref{thm:Cycle_Plus_Triangles}).

\begin{restatable}[Cycle plus triangles;~\citealp{FS92solution}]{theorem}{CyclePlusTriangles}
For any positive integer $p$, let $H$ be any simple graph on $3p$ vertices whose set of edges is the disjoint union of a Hamiltonian cycle and $p$ pairwise vertex-disjoint triangles. Then, $H$ is $3$-colorable.
\label{thm:Cycle_Plus_Triangles}
\end{restatable}

Interestingly, \citet{E90some} conjectured that any ``cycle plus triangles'' graph $H$ always admits a $3$-coloring. This conjecture was proved by \citet{FS92solution}. 

To see how the cycle-plus-triangles'' theorem applies to our problem, consider reindexing the items on the path graph $G$ so that $o_1$ is the most preferred item (recall that the valuations are identical), $o_2$ is the next most preferred item, and so on (ties are broken arbitrarily). Furthermore, for simplicity, let the number of items be a multiple of $3$ (i.e., $m = 3p$); if not, we can add some zero-valued dummy items that will be removed later.

We now construct the graph $H$. For every item $o_1,\dots,o_m$, we add a distinct vertex in $H$. We group the items into triples $\{o_1,o_2,o_3\}, \{o_4,o_5,o_6\}$, and so on, and add edges between every pair of items in each triple, resulting in vertex-disjoint triangles. Thus, for $i \in [p]$, the $i^\text{th}$ triangle contains the $(3i-2)^\textup{th}$, $(3i-1)^\textup{th}$, and $(3i)^\textup{th}$-most preferred items. Further, between any pair of items $o_i$ and $o_j$ that are adjacent along the path graph $G$, we add an edge in $H$. However, if an edge between $o_i$ and $o_j$ already exists from the aforementioned triangle construction, then we connect them via a \emph{nine-vertex gadget}; see \Cref{fig:CyclePlusTriangle}. The gadget helps in avoiding multiedges between vertices by connecting them via a loop of dummy vertices. Finally, we add an edge connecting the extreme items on the path graph to complete the Hamiltonian cycle---again, we do this via a gadget if the items already have an edge between them. Observe that $H$ is a supergraph of $G$.

The graph $H$ constructed above satisfies the conditions of the cycle-plus-triangles theorem~(\Cref{thm:Cycle_Plus_Triangles}) and is therefore $3$-colorable. The set of vertices assigned to each color constitutes an independent set of $H$. By treating each color as an agent, we obtain a feasible and complete allocation.

To argue \EFoneone{}, observe that for each $i \in [p]$, each agent receives exactly one item from its $(3i-2)^\textup{th}$, $(3i-1)^\textup{th}$, and $(3i)^\textup{th}$-most preferred items. Using an item-for-item correspondence argument, it is easy to see that each agent prefers its own bundle over that of any other agent after, if necessary, removing an item in the other agent's bundle from the first triangle (which must be a good) and one of its own items in the last triangle (which must be a chore). Additivity of valuations implies that the allocation is \EFoneone{}. Note that the zero-valued dummy items can be removed without affecting the fairness guarantee.

The above construction for $n=3$ agents can be extended to a general number of agents. To this end, we will use a generalization of the cycle-plus-triangles theorem called the ``cycle plus $n$-cliques'' theorem; the latter applies to graphs consisting of a 
Hamiltonian cycle and vertex disjoint cliques each with at most $n$ vertices. Such graphs are known to always admit an $n$-coloring~\citep{FS97some}. 

\begin{restatable}[Cycle plus $n$-cliques;~\citealp{FS97some}]{theorem}{CyclePlusNCliques}
Let $H$ be any simple graph whose set of edges is the disjoint union of a Hamiltonian cycle and $p$ pairwise vertex-disjoint complete subgraphs each on at most $n$ vertices. Then, $H$ is $n$-colorable.
\label{thm:Cycle_Plus_N_Cliques}
\end{restatable}





In the remainder of this section, we will discuss a procedure to construct a cycle-plus-$n$-cliques graph $H$ from the given fair division instance with $n \geq 3$ agents. We will show that any $n$-coloring in the graph $H$ gives an \EFoneone{} and maximal allocation in the fair division instance. This reduction will establish the desired implication in \Cref{thm:n_Agents_Identical}.

\subsubsection*{Reducing \EFoneone{}+maximal to cycle-plus-$n$-cliques.} 
We are given a conflict graph $G$ that is a path graph. Since the valuations are identical, we can reindex the items so that $o_1$ is the most preferred item, $o_2$ is the next most preferred item (breaking ties arbitrarily), and so on. For convenience, we can make the number of items $m$ an integer multiple of $n$ by adding up to $n-1$ dummy items, each with a value of $0$. These dummy items will be helpful in applying \Cref{thm:Cycle_Plus_N_Cliques} to the reduced graph $H$ and will be removed later. For an illustration of our reduction with $n=3$ agents, see \Cref{fig:CyclePlusTriangle}.

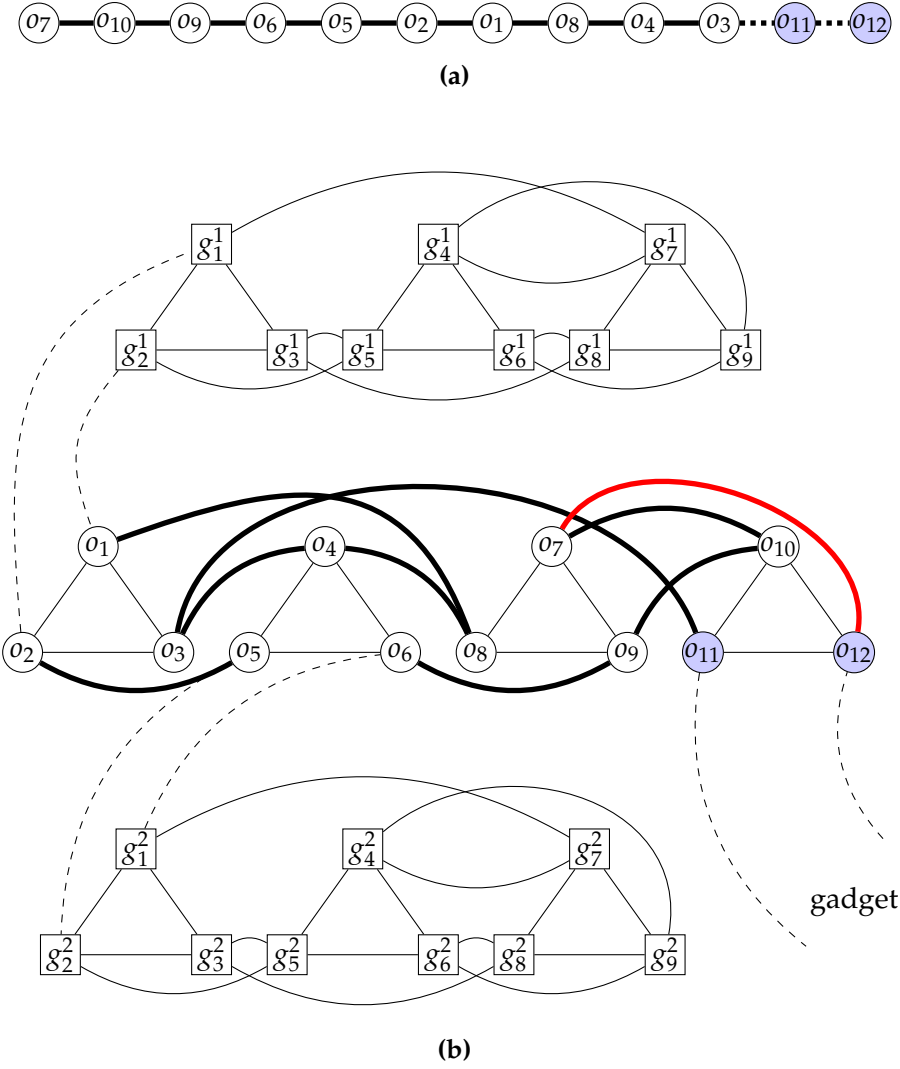
\begin{figure*}[t]
\centering
    \begin{subfigure}[b]{0.9\linewidth}
    \centering
    \begin{tikzpicture}
            \tikzset{mynode/.style = {shape=circle,draw,inner sep=0pt,minimum size=15pt}}
            \node[mynode] (1) at (0,0) {$o_7$};
            \node[mynode] (2) at (1,0) {$o_{10}$};
            \node[mynode] (3) at (2,0) {$o_9$};
            \node[mynode] (4) at (3,0) {$o_6$};
            \node[mynode] (5) at (4,0) {$o_5$};
            \node[mynode] (6) at (5,0) {$o_2$};
            \node[mynode] (7) at (6,0) {$o_1$};
            \node[mynode] (8) at (7,0) {$o_8$};
            \node[mynode] (9) at (8,0) {$o_4$};
            \node[mynode] (10) at (9,0) {$o_3$};
            \node[mynode,fill=blue!20] (11) at (10,0) {$o_{11}$};
            \node[mynode,fill=blue!20] (12) at (11,0) {$o_{12}$};
            \draw[line width=2pt] (1) -- (2);
            \draw[line width=2pt] (2) -- (3);
            \draw[line width=2pt] (3) -- (4);
            \draw[line width=2pt] (4) -- (5);
            \draw[line width=2pt] (5) -- (6);
            \draw[line width=2pt] (6) -- (7);
            \draw[line width=2pt] (7) -- (8);
            \draw[line width=2pt] (8) -- (9);
            \draw[line width=2pt] (9) -- (10);
            \draw[line width=2pt,dotted] (10) -- (11);
            \draw[line width=2pt,dotted] (11) -- (12);
        \end{tikzpicture}
        \caption{}
        \label{fig:CyclePlusTriangle-path}
        \vspace{0.2in}
    \end{subfigure}
    %
    %
    \begin{subfigure}[b]{0.9\linewidth}
    \centering
    \begin{tikzpicture}
            \tikzset{mynode/.style = {shape=circle,draw,inner sep=0pt,minimum size=15pt}}
            \tikzset{gadgetnode/.style = {shape=rectangle,draw,inner sep=0pt,minimum size=15pt}} 
            \node[mynode] (1) at (1,1.4) {$o_1$};
            \node[mynode] (2) at (0,0) {$o_2$};
            \node[mynode] (3) at (2,0) {$o_3$};
            \node[mynode] (4) at (4,1.4) {$o_4$};
            \node[mynode] (5) at (3,0) {$o_5$};
            \node[mynode] (6) at (5,0) {$o_6$};
            \node[mynode] (7) at (7,1.4) {$o_7$};
            \node[mynode] (8) at (6,0) {$o_8$};
            \node[mynode] (9) at (8,0) {$o_9$};
            \node[mynode] (10) at (10,1.4) {$o_{10}$};
            \node[mynode,fill=blue!20] (11) at (9,0) {$o_{11}$};
            \node[mynode,fill=blue!20] (12) at (11,0) {$o_{12}$};
            \draw (1) -- (2);
            \draw (2) -- (3);
            \draw (1) -- (3);
            \draw (4) -- (5);
            \draw (5) -- (6);
            \draw (4) -- (6);
            \draw (7) -- (8);
            \draw (8) -- (9);
            \draw (7) -- (9);
            \draw (10) -- (11);
            \draw (11) -- (12);
            \draw (10) -- (12);
            \draw[line width=2pt] (7) to[bend left] (10);
            \draw[line width=2pt] (10) to[bend right] (9);
            \draw[line width=2pt] (9) to[bend left] (6);
            \draw[line width=2pt] (5) to[bend left] (2);
            \draw[line width=2pt] (1) to[out=20,in=110,looseness=1.3] (8);
            \draw[line width=2pt] (8) to[bend right] (4);
            \draw[line width=2pt] (4) to[bend right] (3);
            \draw[line width=2pt] (3) to[out=80,in=110] (11);
            \draw[line width=2pt,color=red] (12) to[out=80,in=60] (7);
            %
            %
            %
            \def\x{1.5}
            \def\y{4}
            \node[gadgetnode] (13) at (\x+1,\y+1.4) {$g^1_1$};
            \node[gadgetnode] (14) at (\x+0,\y+0) {$g^1_2$};
            \node[gadgetnode] (15) at (\x+2,\y+0) {$g^1_3$};
            \node[gadgetnode] (16) at (\x+4,\y+1.4) {$g^1_4$};
            \node[gadgetnode] (17) at (\x+3,\y+0) {$g^1_5$};
            \node[gadgetnode] (18) at (\x+5,\y+0) {$g^1_6$};
            \node[gadgetnode] (19) at (\x+7,\y+1.4) {$g^1_7$};
            \node[gadgetnode] (20) at (\x+6,\y+0) {$g^1_8$};
            \node[gadgetnode] (21) at (\x+8,\y+0) {$g^1_9$};
            \draw (13) -- (14);
            \draw (14) -- (15);
            \draw (13) -- (15);
            \draw (16) -- (17);
            \draw (17) -- (18);
            \draw (16) -- (18);
            \draw (19) -- (20);
            \draw (20) -- (21);
            \draw (19) -- (21);
            \draw (13) to[bend left] (19);
            \draw (19) to[bend left] (16);
            \draw (16) to[out=40,in=80,looseness=1.2] (21);
            \draw (21) to[bend left] (18);
            \draw (18) to[bend left] (20);
            \draw (20) to[bend left] (15);
            \draw (15) to[bend left] (17);
            \draw (17) to[bend left] (14);
            %
            %
            %
            \def\x{0.5}
            \def\y{-4}
            \node[gadgetnode] (22) at (\x+1,\y+1.4) {$g^2_1$};
            \node[gadgetnode] (23) at (\x+0,\y+0) {$g^2_2$};
            \node[gadgetnode] (24) at (\x+2,\y+0) {$g^2_3$};
            \node[gadgetnode] (25) at (\x+4,\y+1.4) {$g^2_4$};
            \node[gadgetnode] (26) at (\x+3,\y+0) {$g^2_5$};
            \node[gadgetnode] (27) at (\x+5,\y+0) {$g^2_6$};
            \node[gadgetnode] (28) at (\x+7,\y+1.4) {$g^2_7$};
            \node[gadgetnode] (29) at (\x+6,\y+0) {$g^2_8$};
            \node[gadgetnode] (30) at (\x+8,\y+0) {$g^2_9$};
            \draw (22) -- (23);
            \draw (23) -- (24);
            \draw (22) -- (24);
            \draw (25) -- (26);
            \draw (26) -- (27);
            \draw (25) -- (27);
            \draw (28) -- (29);
            \draw (29) -- (30);
            \draw (28) -- (30);
            \draw (22) to[bend left] (28);
            \draw (28) to[bend left] (25);
            \draw (25) to[out=40,in=80,looseness=1.2] (30);
            \draw (30) to[bend left] (27);
            \draw (27) to[bend left] (29);
            \draw (29) to[bend left] (24);
            \draw (24) to[bend left] (26);
            \draw (26) to[bend left] (23);
            %
            %
            \def\x{10.5}
            \def\y{-4}
            \node[] (31) at (\x+1,\y+1.4) {};
            \node[] (32) at (\x+0,\y+0) {};
            \node[] (33) at (\x+0.5,\y+0.7) {gadget};
            %
            %
            %
            \draw[dashed] (2) to[out=95,in=200,looseness=1.2] (13);
            \draw[dashed] (14) to[bend right,looseness=1.2] (1);
            \draw[dashed] (5) to[bend right] (23);
            \draw[dashed] (22) to[bend left] (6);
            \draw[dashed] (11) to[bend right] (32);
            \draw[dashed] (31) to[bend left] (12);
        \end{tikzpicture}
        \caption{}
        \label{fig:CyclePlusTriangle-construction}
    \end{subfigure}
    %
    %
    %
    %
    \caption{An illustration of the reduction in the proof of \Cref{thm:n_Agents_Identical} for $n=3$ agents. Subfigure (a) shows the path graph consisting of ten original items $o_7,o_{10},\dots,o_3$ and two dummy items $o_{11}$ and $o_{12}$ (denoted by shaded nodes). The dummy items can be considered to be an extension of the path graph as shown. Subfigure (b) shows the corresponding cycle-plus-$n$-cliques graph $H$. We divide the items into groups of size $n$, namely $\{o_1,o_2,o_3\}$, $\{o_4,o_5,o_6\}$, and so on such that $\{o_1,o_2,o_3\}$ are the most-preferred $n$ items, $\{o_4,o_5,o_6\}$ are the next most-preferred $n$ items, and so on. The thick black edges simulate the edges of the path graph between items belonging to different groups. For items in the same group that are adjacent along the path (i.e., the \emph{special} pairs, namely, $\{o_1,o_2\}$, $\{o_5,o_6\}$, and $\{o_{11},o_{12}\}$), we simulate the edge via a nine-vertex gadget. The square nodes denote the gadget items. The thin dashed lines denote the connector edges between the special vertices and the gadget vertices. The thick red edge denotes the fictitious edge added between the extreme items on the path to complete the Hamiltonian cycle. \Cref{fig:CyclePlusTriangle-Coloring} shows a $3$-coloring of the graph $H$.}
\label{fig:CyclePlusTriangle}
\end{figure*}

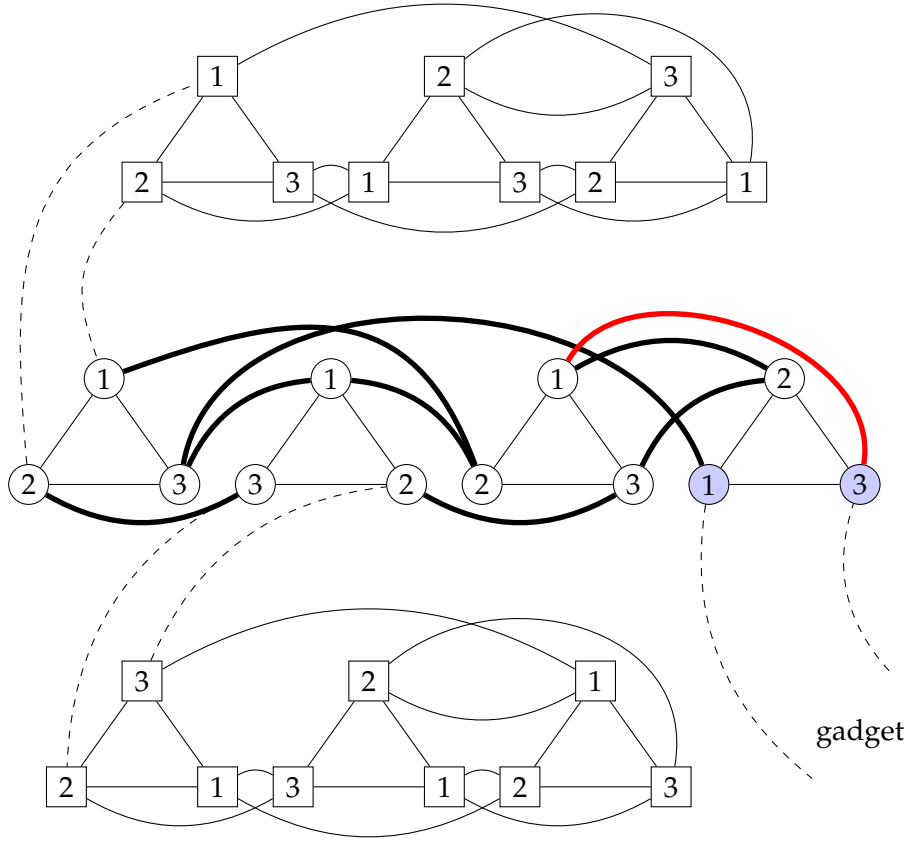
\begin{figure*}[t]
\centering
    \begin{tikzpicture}
            \tikzset{mynode/.style = {shape=circle,draw,inner sep=0pt,minimum size=15pt}}
            \tikzset{gadgetnode/.style = {shape=rectangle,draw,inner sep=0pt,minimum size=15pt}} 
            \node[mynode] (1) at (1,1.4) {$1$};
            \node[mynode] (2) at (0,0) {$2$};
            \node[mynode] (3) at (2,0) {$3$};
            \node[mynode] (4) at (4,1.4) {$1$};
            \node[mynode] (5) at (3,0) {$3$};
            \node[mynode] (6) at (5,0) {$2$};
            \node[mynode] (7) at (7,1.4) {$1$};
            \node[mynode] (8) at (6,0) {$2$};
            \node[mynode] (9) at (8,0) {$3$};
            \node[mynode] (10) at (10,1.4) {$2$};
            \node[mynode,fill=blue!20] (11) at (9,0) {$1$};
            \node[mynode,fill=blue!20] (12) at (11,0) {$3$};
            \draw (1) -- (2);
            \draw (2) -- (3);
            \draw (1) -- (3);
            \draw (4) -- (5);
            \draw (5) -- (6);
            \draw (4) -- (6);
            \draw (7) -- (8);
            \draw (8) -- (9);
            \draw (7) -- (9);
            \draw (10) -- (11);
            \draw (11) -- (12);
            \draw (10) -- (12);
            \draw[line width=2pt] (7) to[bend left] (10);
            \draw[line width=2pt] (10) to[bend right] (9);
            \draw[line width=2pt] (9) to[bend left] (6);
            \draw[line width=2pt] (5) to[bend left] (2);
            \draw[line width=2pt] (1) to[out=20,in=110,looseness=1.3] (8);
            \draw[line width=2pt] (8) to[bend right] (4);
            \draw[line width=2pt] (4) to[bend right] (3);
            \draw[line width=2pt] (3) to[out=80,in=110] (11);
            \draw[line width=2pt,color=red] (12) to[out=80,in=60] (7);
            %
            %
            %
            \def\x{1.5}
            \def\y{4}
            \node[gadgetnode] (13) at (\x+1,\y+1.4) {$1$};
            \node[gadgetnode] (14) at (\x+0,\y+0) {$2$};
            \node[gadgetnode] (15) at (\x+2,\y+0) {$3$};
            \node[gadgetnode] (16) at (\x+4,\y+1.4) {$2$};
            \node[gadgetnode] (17) at (\x+3,\y+0) {$1$};
            \node[gadgetnode] (18) at (\x+5,\y+0) {$3$};
            \node[gadgetnode] (19) at (\x+7,\y+1.4) {$3$};
            \node[gadgetnode] (20) at (\x+6,\y+0) {$2$};
            \node[gadgetnode] (21) at (\x+8,\y+0) {$1$};
            \draw (13) -- (14);
            \draw (14) -- (15);
            \draw (13) -- (15);
            \draw (16) -- (17);
            \draw (17) -- (18);
            \draw (16) -- (18);
            \draw (19) -- (20);
            \draw (20) -- (21);
            \draw (19) -- (21);
            \draw (13) to[bend left] (19);
            \draw (19) to[bend left] (16);
            \draw (16) to[out=40,in=80,looseness=1.2] (21);
            \draw (21) to[bend left] (18);
            \draw (18) to[bend left] (20);
            \draw (20) to[bend left] (15);
            \draw (15) to[bend left] (17);
            \draw (17) to[bend left] (14);
            %
            %
            %
            \def\x{0.5}
            \def\y{-4}
            \node[gadgetnode] (22) at (\x+1,\y+1.4) {$3$};
            \node[gadgetnode] (23) at (\x+0,\y+0) {$2$};
            \node[gadgetnode] (24) at (\x+2,\y+0) {$1$};
            \node[gadgetnode] (25) at (\x+4,\y+1.4) {$2$};
            \node[gadgetnode] (26) at (\x+3,\y+0) {$3$};
            \node[gadgetnode] (27) at (\x+5,\y+0) {$1$};
            \node[gadgetnode] (28) at (\x+7,\y+1.4) {$1$};
            \node[gadgetnode] (29) at (\x+6,\y+0) {$2$};
            \node[gadgetnode] (30) at (\x+8,\y+0) {$3$};
            \draw (22) -- (23);
            \draw (23) -- (24);
            \draw (22) -- (24);
            \draw (25) -- (26);
            \draw (26) -- (27);
            \draw (25) -- (27);
            \draw (28) -- (29);
            \draw (29) -- (30);
            \draw (28) -- (30);
            \draw (22) to[bend left] (28);
            \draw (28) to[bend left] (25);
            \draw (25) to[out=40,in=80,looseness=1.2] (30);
            \draw (30) to[bend left] (27);
            \draw (27) to[bend left] (29);
            \draw (29) to[bend left] (24);
            \draw (24) to[bend left] (26);
            \draw (26) to[bend left] (23);
            %
            %
            \def\x{10.5}
            \def\y{-4}
            \node[] (31) at (\x+1,\y+1.4) {};
            \node[] (32) at (\x+0,\y+0) {};
            \node[] (33) at (\x+0.5,\y+0.7) {gadget};
            %
            %
            %
            \draw[dashed] (2) to[out=95,in=200,looseness=1.2] (13);
            \draw[dashed] (14) to[bend right,looseness=1.2] (1);
            \draw[dashed] (5) to[bend right] (23);
            \draw[dashed] (22) to[bend left] (6);
            \draw[dashed] (11) to[bend right] (32);
            \draw[dashed] (31) to[bend left] (12);
        \end{tikzpicture}
        \caption{A $3$-coloring of the graph $H$ shown in \Cref{fig:CyclePlusTriangle}.}
\label{fig:CyclePlusTriangle-Coloring}
\end{figure*}

Let $m$ denote the total number of items in the path graph after adding the dummy items. The cycle-plus-$n$-cliques graph $H$ is constructed as follows: For each item, we will create a distinct vertex in $H$. To define the edges in $H$, the items are grouped into $m/n$ groups, formed by putting together blocks of $n$ most-preferred items as per agents' preferences. That is, the groups are $\{o_1,o_2,\dots,o_n\}$, $\{o_{n+1},o_{n+2},\dots,o_{2n}\}$, and so on. Each group corresponds to an $n$-clique in $H$, meaning we will add an edge between every pair of vertices corresponding to the items in the same group. In \Cref{fig:CyclePlusTriangle}, this construction is illustrated in the form of four triangles $\{o_1,o_2,o_3\},\{o_4,o_5,o_6\},\{o_7,o_8,o_9\}$ and $\{o_{10},o_{11},o_{12}\}$ in subfigure~(b).

Next, we will incorporate the edges in the path graph $G$ into the cycle-plus-$n$-cliques graph $H$. For any pair of items that are adjacent in $G$, we will add an edge between the corresponding vertices in $H$ (shown via thick black edges in \Cref{fig:CyclePlusTriangle}). 
However, there may be items that are adjacent in the path graph $G$ that belong to the same group; we call such pairs \emph{special}. In \Cref{fig:CyclePlusTriangle}, the pairs $\{o_1,o_2\}$, $\{o_5,o_6\}$, and $\{o_{11},o_{12}\}$ correspond to the special pairs. To connect the vertices corresponding to any special pair of items, we use a \emph{gadget} (\Cref{fig:CyclePlusTriangle} shows two such gadgets). For any $i$, the $i^\textup{th}$ gadget consists of nine vertices $g^i_1,g^i_2,\dots,g^i_9$. These vertices are connected as follows: First, we group the gadget vertices into three triples $\{g^i_1,g^i_2,g^i_3\},\{g^i_4,g^i_5,g^i_6\}$, and $\{g^i_7,g^i_8,g^i_9\}$ and connect all vertex pairs within the same triple. Next, we add eight more edges, namely $\{g^i_1,g^i_7\}$, $\{g^i_7,g^i_4\}$, $\{g^i_4,g^i_9\}$, $\{g^i_9,g^i_6\}$, $\{g^i_6,g^i_8\}$, $\{g^i_8,g^i_3\}$, $\{g^i_3,g^i_5\}$, and $\{g^i_5,g^i_2\}$. Note that with this construction, the subgraph induced by the gadget consists of a union of pairwise vertex-disjoint triangles and a Hamiltonian cycle except for the missing edge between $g^i_1$ and $g^i_2$. We connect the gadget vertices $g^i_1$ and $g^i_2$ with the vertices corresponding to the $i^\textup{th}$ special pair of items using two \emph{connector} edges (shown in \Cref{fig:CyclePlusTriangle}, subfigure (b) via thin dashed lines). 
It is relevant to note that for general $n$, while the original items will induce $n$-cliques, the gadget will nevertheless consist only of triangles.

Finally, we add one more edge between the vertices corresponding to the start and the end items of the path (including the dummy items). In \Cref{fig:CyclePlusTriangle}, this edge is shown via a thick red edge between $o_7$ and $o_{12}$. (Again, if an edge already exists between these items, we connect them through another gadget.) This finishes the construction of the graph~$H$.

\begin{proof} (of \Cref{thm:n_Agents_Identical})
    It is easy to see that the edge set of the graph $H$ consists of a disjoint union of pairwise vertex-disjoint complete subgraphs of size at most $n$ and a Hamiltonian cycle. By the cycle-plus-$n$-cliques theorem~(\Cref{thm:Cycle_Plus_N_Cliques}), the graph $H$ admits an $n$-coloring. Note that since $H$ is a supergraph of $G$, an $n$ coloring of $H$ induces an $n$-coloring of $G$.

    Consider any bijection between colors and agents. Assign each item in $G$ (including dummy items) to the agent whose associated color is assigned to the corresponding vertex. Since the coloring is proper, no pair of adjacent vertices is assigned to the same color; thus, the assignment is feasible. Furthermore, since all items are allocated, the allocation is complete and, therefore, maximal.

    To see why the allocation is \EFoneone{}, observe that each agent receives exactly one item from the vertices corresponding to each $n$-clique. Fix a pair of agents $h$ and $k$. For any $t \in \{1,2,\dots,m/n\}$, denote the items in group $t$ as $\{o_{(t-1)n + 1},\dots,o_{(t-1)n + n}\}$. Observe that for any $t \in \{1,2,\dots,m/n - 1\}$, agent $h$ prefers the item it receives from the $t^\textup{th}$ group over that received by agent $k$ from the $(t+1)^\textup{th}$ group. By additivity, it follows that agent $h$ prefers its bundle over agent $k$'s bundle, after the removal of its group $t$ item and agent $k$'s group $1$ item. Furthermore, the group $t$ item removed from agent $h$'s bundle must be a chore, and the group $1$ item removed from agent $k$'s bundle must be a good; otherwise, these items can be included in the comparison of bundles without affecting envy-freeness. Thus, the allocation is envy-free up to the removal of one good and one chore (\EFoneone{}), as desired.
    
    Finally, observe that the vertices corresponding to all dummy items are located in the same $n$-clique. Since all vertices within each $n$-clique are assigned distinct colors, no agent receives more than one dummy item. As the dummy items have zero value, their removal does not affect the \EFoneone{} guarantee. Furthermore, for monotone valuations, \EFoneone{} is equivalent to \EF{1}. This proves \Cref{thm:n_Agents_Identical}.
\end{proof}

We conclude this subsection with a discussion about the computational aspects. The proof of \Cref{thm:n_Agents_Identical} provides an efficient (i.e., polynomial-sized) reduction from the fair division problem to the cycle-plus-$n$-cliques problem. Thus, a computationally efficient (i.e., polynomial-time) algorithm for computing the $n$-coloring in the latter problem would provide a computationally efficient algorithm for computing an \EFoneone{} and maximal allocation of the path graph. However, showing that the cycle-plus-$n$-cliques problem admits a polynomial-time algorithm remains an open problem for all $n \geq 3$. 

The following related problem, called ``Fair Independent Set in Cycle (\FairIS{})'', has been studied by \citet{AAB+17fair} and \citet{H22complexity}.

\FairIS{}: Given a cycle $C$ on $n$ vertices and a partition $V_1,V_2,\dots,V_t$ of its vertex set into $t$ sets, find an independent set $S$ of $C$ such that $|S \cap V_i| \geq \frac{1}{2} |V_i| - 1$ for all $i \in [t]$.

\citet{AAB+17fair} established that every instance of \FairIS{} is guaranteed to have the desired independent set. Furthermore, they showed that when the vertex sets $V_1,\dots,V_t$ are each of size $3$, there exist two disjoint independent sets, each containing a vertex from each of the $V_i$'s. The latter result implies the existence of a $3$-coloring in the cycle-plus-triangles problem.

Following this, \citet{H22complexity} showed that the \FairIS{} problem is \PPAC{}. However, this work left open the problem of determining the precise computational complexity of \FairIS{} restricted to vertex sets of equal size, as is the case in the cycle-plus-$n$-cliques problem. Note that even with equi-sized vertex sets, the cycle-plus-$n$-cliques problem is at least as hard as \FairIS{}, since the desired independent set in the latter problem can be inferred from a valid $n$-coloring of the former. We refer the reader to the survey by \citet{A93restricted} and the references therein for further details.

\subsection{Uniform Valuations on Trees}
\label{sec:uniform}

We consider the case of uniform valuations, i.e., $v_i(S) = \left|S\right|$ for all $S \subseteq M$. For such valuations, our problem is closely related to an \emph{equitable coloring} of a graph, where the number of vertices assigned to each color differs by at most $1$. The key difference is that an  \emph{equitable coloring} must be complete---no vertex can remain uncolored---whereas in our problem, we aim to find a \emph{maximal equitable (partial) coloring} defined as follows.  


Given a graph $G = (V, E)$ and $n \in \N$, a \emph{maximal equitable $n$-coloring} is a subpartition $(S_1, \dots, S_n)$ of $V$, where $S_1, \dots, S_n \subseteq V$, with the following conditions:
    \begin{itemize}
        \item $S_1, \dots, S_n$ are disjoint.
        \item (independence) For any pair of distinct vertices $x, y \in S_i$, the vertices $x, y$ are not adjacent.
        \item (maximality) For every $v \in V \setminus (S_1 \cup \dots \cup S_n)$ and $i \in \{1, \dots, n\}$, $v$ has an adjacent vertex in $S_i$. 
        \item (equitability) $\max_{i \in [n]} \left|S_i\right| - \min_{i \in [n]} \left|S_i\right| \leq 1$.
    \end{itemize}

Below, we prove that a maximal equitable $n$-coloring exists when $G$ is a tree. It remains an open problem whether this can be extended to all graphs.

\begin{theorem}
For every tree $T = (V, E)$ and every $n \in \N$, there exists a maximal equitable $n$-coloring of $T$.
    \label{thm:uni-tree2}
\end{theorem}

\begin{proof}
    Consider $T$ as a rooted tree with root $r$. We prove a stronger fact that there exists a maximal equitable $n$-coloring of $T$ such that $r$ is colored with a \emph{higher color} or is uncolored. Here, for a coloring $\mathcal{S} = (S_1, \dots, S_n)$, we define that color $i$ is a \emph{higher color} if $|S_i| = \max_{j \in [n]} |S_j|$.

    We will prove this by strong induction on the number of vertices of $T 
= (V, E)$. The base case $|V| = 1$ is trivial.

    For the inductive step, consider any tree $T 
= (V, E)$ with $|V| > 1$ and suppose that 
the statement holds
for all trees with smaller number of vertices.
    Let $r_1, \dots, r_k$ be the children of $r$, and let $T_1, \dots, T_k$ be the subtree of $r_1, \dots, r_k$, respectively. By the inductive hypothesis, 
    for each $i$, there exists a maximal equitable $n$-coloring of $T_i$ where $r_i$ is a higher color or is uncolored; let it be $\mathcal{S}^{(i)} = (S^{(i)}_1, \dots, S^{(i)}_n)$. Without loss of generality, we assume $|S^{(i)}_1| \geq \dots \geq |S^{(i)}_n|$. We say that $T_i$ is a \emph{singular subtree} if $\mathcal{S}^{(i)}$ has only one higher color (color $1$) and $r_i$ has color $1$. Without loss of generality, for some $c$, $T_1, \dots, T_c$ are singular subtrees and $T_{c+1}, \dots, T_k$ are not.
    
    First, we aim to obtain an equitable $n$-coloring $\mathcal{S}$ of $T_1 \cup \dots \cup T_k$ in the following way:
    \begin{enumerate}
        \item Initialize $\mathcal{S} = (\emptyset, \dots, \emptyset)$ and $x = 0$.
        \item For $i = 1, \dots, k$, do the following:
        \begin{enumerate}
            \item $(S_1, \dots, S_n) \gets (S_1 \cup S^{(i)}_{n-x+1}, \dots, S_x \cup S^{(i)}_n, S_{x+1} \cup S^{(i)}_1, \dots, S_n \cup S^{(i)}_{n-x})$ where the indices of $\mathcal{S}^{(i)}$ wrap around to 1 after $n$, and,
            \item $x \gets (x + (\text{number of higher colors in $\mathcal{S}^{(i)}$})) \bmod n$.
        \end{enumerate}
    \end{enumerate}
    It is not difficult to see that, for every loop just after step 2 (b), colors $1, \dots, x$ are higher colors of $\mathcal{S}$, and colors $x+1, \dots, n$ are not (or, when $x = 0$, all colors are higher colors). Also, observe that the algorithm works in a way that $r_1, \dots, r_{\min(c, n)}$ will have colors $1, \cdots, \min(c, n)$, respectively, as $T_1, \dots, T_{\min(c, n)}$ are singular subtrees.
    
    The problem is that, in order to obtain a coloring of $T$, we need to consider vertex $r$. We divide it into two cases to solve this issue.
    \begin{itemize}
        \item (Case 1. $c \geq n$) $r_1, \dots, r_n$ will have colors $1, \dots, n$, respectively. Therefore, $r$ can be left uncolored, and $\mathcal{S}$ is already a maximum equitable $n$-coloring of $T$.
        \item (Case 2. $c < n$) We aim to color vertex $r$ with color $n$. The case when this cannot be done is that some $r_i$ already has color $n$. However, since $r_1, \dots, r_c$ have colors $1, \dots, c \ (< n)$, we know that $T_i$ is not a singular subtree. So, there is another higher color inside $T_i$, say color $x$ (in the context of coloring $\mathcal{S}$; $x \neq n$ must hold). Then, we can swap color $x$ and color $n$ inside the entire $T_i$ --- then $r_i$ no longer has color $n$, and the number of vertices of each color does not change. We repeat this process until no $r_i$ has color $n$. Since color $n$ is used by the least number of vertices, coloring vertex $r$ in color $n$ does not break the equitability condition, and color $n$ becomes a higher color.
    \end{itemize}
    
    Therefore, we obtained the desired coloring of $T$.
\end{proof}

\section{Concluding Remarks}
\label{sec:Conclusion}

We studied the problem of fair and efficient allocation of indivisible items under conflict constraints. Although techniques from the unconstrained setting are not applicable, we demonstrated, using a color-switching technique, that an \EF{1} and maximal allocation always exists for two agents with monotone valuations for any conflict graph. Resolving the existence question for three or more agents with monotone valuations for path graphs---and, more generally, for interval graphs---is an exciting open problem. Another relevant question is to determine if there is a polynomial-time algorithm for the two-agent problem under monotone valuations and arbitrary conflict graphs. It would also be interesting to consider allocations that ``minimize wastage'' (i.e., leave the fewest items unassigned) or satisfy other notions of fairness such as proportionality, equitability, or maximin fair share. 


Additionally, we applied a novel approach using the cycle-plus-$n$-cliques theorem to demonstrate the existence of an \EFoneone{} and maximal allocation for identical additive valuations on a path graph. However, this result does not provide an efficient algorithm for computing such an allocation. An interesting problem for future research is to determine whether there is a polynomial-time algorithm for finding an \EFoneone{} and maximal allocation on a path graph when there are $n \geq 3$ agents.

\FloatBarrier

\section*{Acknowledgments}
We are thankful to anonymous reviewers of AAMAS 2024 and IJCAI 2025 for their helpful feedback, and to Amit Kumar for valuable discussions during the early stages of this work. AI acknowledges support from JST FOREST Grant Numbers JPMJPR20C1. SN gratefully acknowledges the support of a MATRICS grant (MTR/2021/000367) and a Core Research Grant (CRG/2022/009169) from SERB, Govt. of India, a TCAAI grant (DO/2021-TCAI002-009), and a TCS grant (MOU/CS/10001981-1/22-23). RS acknowledges support of the Department of Atomic Energy, Government  of India, under project no. RTI4014. RV acknowledges support from DST INSPIRE grant no. DST/INSPIRE/04/2020/000107, SERB/ANRF grant no. CRG/2022/002621, and Mr. D.P. Gupta Chair Professorship.

\FloatBarrier

\bibliographystyle{plainnat}
\bibliography{abb,References,main}

\end{document}